\documentclass[journal]{IEEEtran}

\usepackage{algorithmic}
\usepackage{array}
\usepackage{stfloats}
\usepackage{url}
\usepackage{verbatim}
\usepackage{cite}
\usepackage{amsmath,amssymb,amsfonts}
\usepackage{textcomp}
\usepackage{xcolor}
\usepackage[caption=false,font=footnotesize,labelfont=rm,textfont=rm]{subfig}

\usepackage{bbding}
\usepackage{bm}
\usepackage{graphicx}
\usepackage{amsthm}
\usepackage{algorithm}
\usepackage[utf8]{inputenc}

\floatname{algorithm}{Algorithm}
\newtheorem{theorem}{\textbf{Theorem}}

\newtheorem{condition}{\textbf{Condition}}
\newtheorem{lemma}{\textbf{Lemma}}
\newtheorem{corollary}{\textbf{Corollary}}

\begin{document}

\title{Design and Optimization of Hierarchical Gradient Coding for Distributed Learning at Edge Devices}

\author{Weiheng Tang, Jingyi Li, Lin Chen,~\IEEEmembership{Member,~IEEE},\\and Xu Chen,~\IEEEmembership{Senior Member,~IEEE}
\thanks{Weiheng Tang, Jingyi Li and Xu Chen are with the School of Computer Science and
Engineering, Sun Yat-sen University, Guangzhou 510275, China (E-mail:
\{tangwh6, lijy573\}@mail2.sysu.edu.cn; chenxu35@mail.sysu.edu.cn).}
\thanks{Lin Chen is with the School of Computer Science and Engineering, Guang-Dong Province Key Laboratory of Information Security Technology, Sun Yat-sen University, Guangzhou 510275, China (E-mail: chenlin69@mail.sysu.edu.cn).}}
\maketitle



\begin{abstract}
Edge computing has recently emerged as a promising paradigm to boost the performance of distributed learning by leveraging the distributed resources at edge nodes. Architecturally, the introduction of edge nodes adds an additional intermediate layer between the master and workers in the original distributed learning systems, potentially leading to more severe straggler effect. Recently, coding theory-based approaches have been proposed for stragglers mitigation in distributed learning, but the majority focus on the conventional workers-master architecture. In this paper, along a different line, we investigate the problem of mitigating the straggler effect in hierarchical distributed learning systems with an additional layer composed of edge nodes. Technically, we first derive the fundamental trade-off between the computational loads of workers and the stragglers tolerance. Then, we propose a hierarchical gradient coding framework, which provides better stragglers mitigation, to achieve the derived computational trade-off. To further improve the performance of our framework in heterogeneous scenarios, we formulate an optimization problem with the objective of minimizing the expected execution time for each iteration in the learning process. We develop an efficient algorithm to mathematically solve the problem by outputting the optimum strategy. Extensive simulation results demonstrate the superiority of our schemes compared with conventional solutions.
\end{abstract}


\section{Introduction}\label{set1}
\IEEEPARstart{M}{achine} learning has become one of the most common computing tasks due to the proliferation of artificial intelligence. However, it is challenging to run large-scale learning tasks on a single machine with limited computation and storage capacities given the increasing sizes of models and datasets. Hence, distributed learning has emerged as an effective paradigm for large-scale learning tasks recently\cite{2011Parallelized,2011HOGWILD,boyd2011distributed,li2018learning}. A complex learning task is partitioned into distinct sub-tasks and executed by several workers simultaneously. In this way, the computational loads of the central master is distributed to the workers. Meanwhile, more and more wireless edge devices with decent communication and computation capabilities have been utilized as workers to execute the learning sub-tasks\cite{li2018learning,chen2019deep}.

However, conventional workers-master architecture for distributed learning consumes significant resources such as network bandwidth at the central master, which may cause severe bottlenecks in learning. To mitigate this side effect, edge computing\cite{shi2016edge,zhou2019edge,xu2021edge} is applied by exploiting edge nodes (e.g., base stations with MEC servers and smart IoT gateways) located in the proximity of distributed workers for enhancing computation and communication capabilities. By leveraging edge nodes to preliminary process the results from workers (i.e., edge devices) in a distributed manner, edge computing can greatly relieve communication burden of the central master and utilize the distributed resources at edge devices more sufficiently. For instance, if a central master communicates with 100 workers, adding 10 edge nodes for pre-aggregation would reduce the communication loads on the master by a factor of 10 compared to the original setup. Due to its profound benefits, such hierarchical workers-edge-master architecture has been introduced into many distributed learning systems\cite{ren2019survey,li2018learning,chen2019deep,hiertrain9094236,HEFL9127160,InFL9831099}.

Although distributed learning has emerged as a promising learning paradigm and has been adopted in a multitude of smart wireless and IoT applications, it also brings new design challenges. Specifically, individual edge workers in distributed learning systems, especially in the wireless setting, may take a long time or even fail to complete their computing tasks due to their weak computation and communication capabilities\cite{lee2017speeding}. These workers are called \textit{stragglers}. The unpredictable delay caused by stragglers, namely \textit{stragglers effect}, greatly degrades the performance and even outweighs the benefits of distributed learning.

Recently, coding theory has been employed as a powerful tool for stragglers mitigation. Specifically, the combination of distributed gradient descent and coding techniques is called \textit{gradient coding}, which was first introduced in \cite{tandon2017gradient}, and inspired a series of follow-up works\cite{raviv2018gradient,halbawi2018improving,wang2019heterogeneity,DSW8755563,TCOM9755943}. By adding redundant computing tasks to the workers, gradient coding enables the central master to recover the computing results with a subset of workers, thus reducing the stragglers effect.

Most of the studies on gradient coding have focused on the conventional workers-master topology design so far (cf. Section \ref{set7} for a synthetic review of related works and Table \ref{tab2} summarizing the contribution of our work compared to the state of the art). The design and optimization of gradient coding schemes in the hierarchical distributed learning systems at the edge is still largely unexplored territory, where there are two main technical challenges.
\begin{itemize}
    \item First, hierarchical distributed learning systems may potentially suffer from more severe straggler effect. Except for those straggling workers, some edge nodes may become edge stragglers, because of the communication workloads through layers and their weak communication capabilities with the master or workers. All the workers connected to such an edge straggler will also be implicated to become stragglers, thus calling for a novel design of gradient coding scheme taking into account the hierarchical architecture.
    \item Second, hierarchical distributed learning systems are more heterogeneous than the conventional master-worker architecture. Aside from workers with diverse communication and computation capabilities, heterogeneous edge nodes must be considered. Such heterogeneity needs to be carefully handled in the optimization of the gradient coding scheme, which underscores the importance of selecting a suitable stragglers tolerance level to improve performance.
\end{itemize}

Given the challenges above, in this paper, we embark on developing a novel hierarchical gradient coding framework for hierarchical distributed learning systems, which provides better stragglers mitigation for more severe stragglers effect in hierarchical architecture. Our proposed coding scheme encodes the partial gradient results of each worker with two layers of codes and utilizes the decoding capabilities of both the master and edge nodes. To further improve the performance of our proposed framework in heterogeneous scenarios, we formulate an optimization problem with the objective of minimizing the expected execution time for each iteration in hierarchical coded distributed learning systems, and develop an optimization algorithm with theoretical performance bound. The main contributions of our work are summarized as follows:
\begin{itemize}
\item We derive a fundamental computational trade-off between the computational loads of workers and the stragglers tolerance which also depicts the relationship between the added computing redundancy and stragglers mitigation capability of hierarchical distributed learning systems.

\item We propose a hierarchical gradient coding framework to achieve the optimal trade-off aforementioned. Our proposed coding scheme encodes the partial gradient results with two layers of codes and utilizes the decoding capabilities of both the master and edge nodes. Hence, it can achieve the same stragglers tolerance as the conventional gradient coding schemes with fewer computational loads.

\item We model the total runtime of each iteration in the proposed hierarchical coded distributed learning system. Accordingly, we further formulate a jointly node and coding scheme selection problem (JNCSS) with the objective of minimizing the expected execution time for each iteration. Then. we develop an optimization algorithm to solve the problem by outputting the optimum nodes and coding scheme selection. We also give the theoretical performance bound of our proposed algorithm.

\item We conduct extensive numerical experiments on popular image classification datasets, which demonstrate the superiority of our proposed schemes, compared to several conventional schemes. For example, our proposed hierarchical gradient coding scheme achieves up to $60.1\%$ and $59.8\%$ performance gain over the conventional gradient coding and uncoded schemes in terms of runtime metric, respectively.
\end{itemize}

\textbf{Roadmap}. This paper is organized as follows. We first  present the system model and computational trade-off in Section \ref{set2}. And then, we propose a hierarchical gradient coding scheme in Section \ref{set3}. In Section \ref{set5}, we first model the total runtime of each iteration in the proposed hierarchical coded distributed learning system. After that, we analyze and optimize the total runtime of each iteration based on our modeling. In Section \ref{set6}, numerical experiments is carried out. And in Section \ref{set7}, we present the related works about stragglers mitigation in distributed systems. Finally, the conclusions are drawn in Section \ref{set8}.

\section{System Model and Computational Trade-off}\label{set2}
In this section, we first describe the generic model of hierarchical distributed learning systems. We then derive the fundamental computational trade-off between the computational loads of workers and the stragglers tolerance.

\textbf{Notation}. For any $x,y\in\mathbb{N}^{+}$, $\boldsymbol{1}_{x\times y}$ and $\boldsymbol{0}_{x\times y}$ denote the matrix of size $x\times y$ with all entries being 1 and 0, respectively. For convenience, we use $[x]$ to denote the set $\{1,2,...,x\}$, and use $[0\colon x)$ to denote the set $\{0,1,...,x-1\}$.
\subsection{System Model}\label{set21}
\begin{figure}[t]
	\centering
	\includegraphics[width=0.5\textwidth]{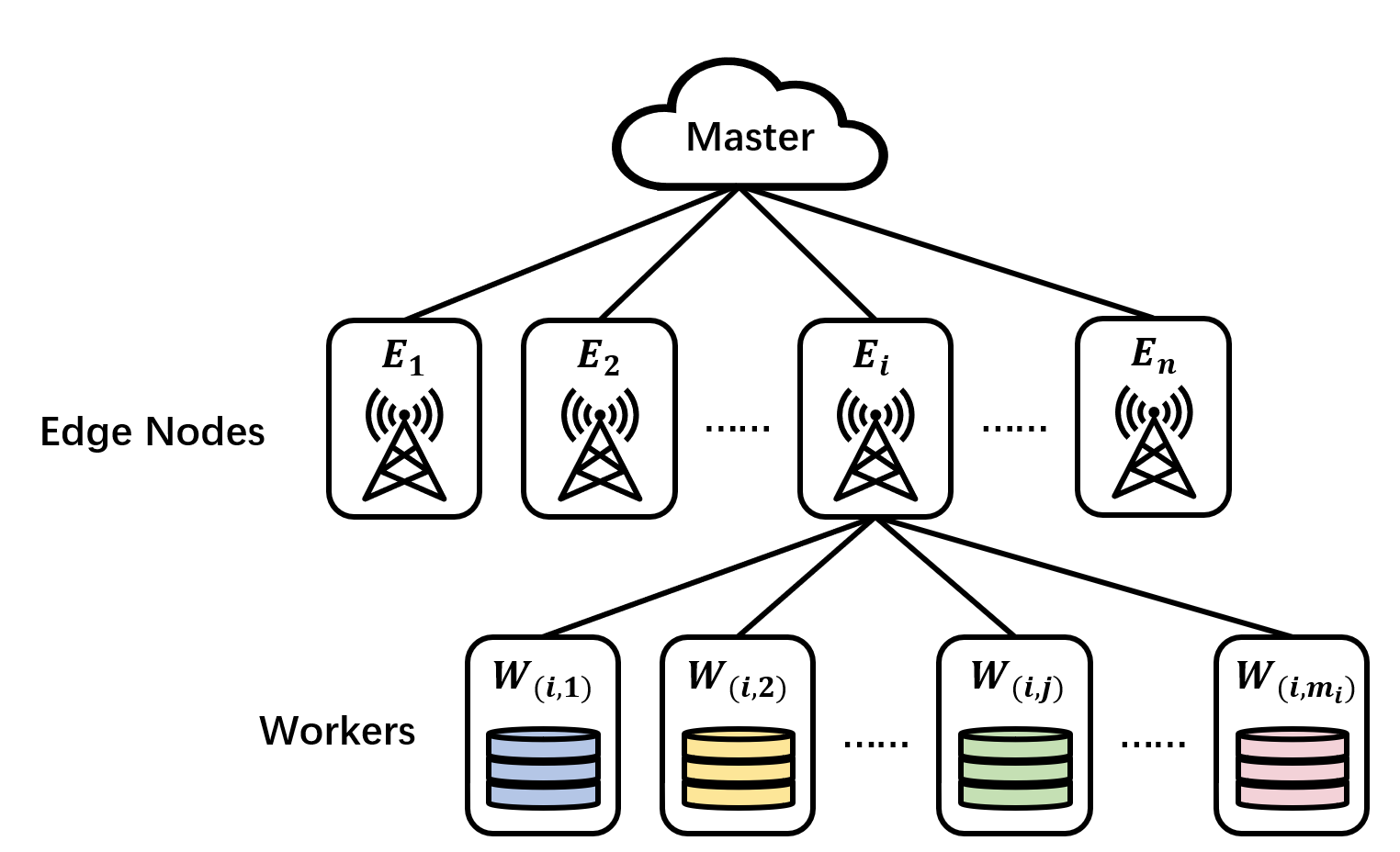}
	\caption{Illustration of a hierarchical distributed learning system.}
 \label{fig1}
\end{figure}
We consider a hierarchical distributed learning system with $n$ edge nodes $E_{1},E_{2},...,E_{n}$, where, for any $i\in[n]$, the edge node $E_{i}$ connects to $m_{i}$ workers $W_{(i,1)},W_{(i,2)},...,W_{(i,m_{i})}$ separately and all of them interact with a same central master, as illustrated in Fig. \ref{fig1}.

Such a hierarchical system trains a machine learning model with distributed gradient descent method by crowdsourcing the training tasks and allocating the sub-datasets among the workers which compute the partial gradient results in parallel. The central master will recover the full gradient with its received results during each training iteration. More precisely, we intend to learn model parameters $\beta$ by minimizing the loss function $L(\mathcal{D};\beta)$, where $\mathcal{D}=\{(x_{i},y_{i})\}^{N}_{i=1}$ is a given dataset. And the full gradient of the loss function can be expressed as
\begin{equation}\label{1}
    \nabla L(\mathcal{D};\beta)=\sum^{N}_{i=1}\nabla l(x_{i},y_{i};\beta).
\end{equation}

All of the data will be divided into $K$ disjoint sub-datasets $\{\mathcal{D}_{1},\mathcal{D}_{2},...,\mathcal{D}_{K}\}$. Hence the full gradient of the loss function can also be expressed as
\begin{equation}\label{2}
    \nabla L(\mathcal{D};\beta)=\sum^{K}_{i=1}\sum_{(x,y)\in\mathcal{D}_{i}}\nabla l(x,y;\beta).
\end{equation}

Next, we introduce the concept of stragglers of workers and edge nodes, as illustrated in Fig. \ref{fig1-3}. Due to limited computation or communication capabilities, certain workers and edge nodes may require a long time to return their results, hence becoming stragglers. Specifically, edge stragglers will result in all the workers connected to them also becoming stragglers.

Considering the load balance of workers, each worker is supposed to train $D$ sub-datasets, i.e. computational loads, according to a fixed allocation scheme\footnote{Due to space limit, we will consider the more involved cases of unbalanced trarining workload allocation in a future work.}. During each iteration, workers will compute the partial gradients and send the linear combinations of them to edge nodes. For any edge node $E_{i}$, which plays the role of partial aggregator and decoder, it will also combine its received information linearly and return the results to the central master after receiving from the fastest $f_{w}^{i}$ workers. We aim to recover the full gradient from these linear combinations received from the fastest $f_{e}$ edge nodes with designed coding scheme. Equivalently, we would like to guarantee the stragglers tolerance level such that the central master is able to tolerate any $s_{e}=n-f_{e}$ stragglers and each edge node is able to tolerate any $s_{w}=m_{i}-f_{w}^{i}$ stragglers.
\begin{figure}[t]
	\centering
	\includegraphics[width=0.5\textwidth]{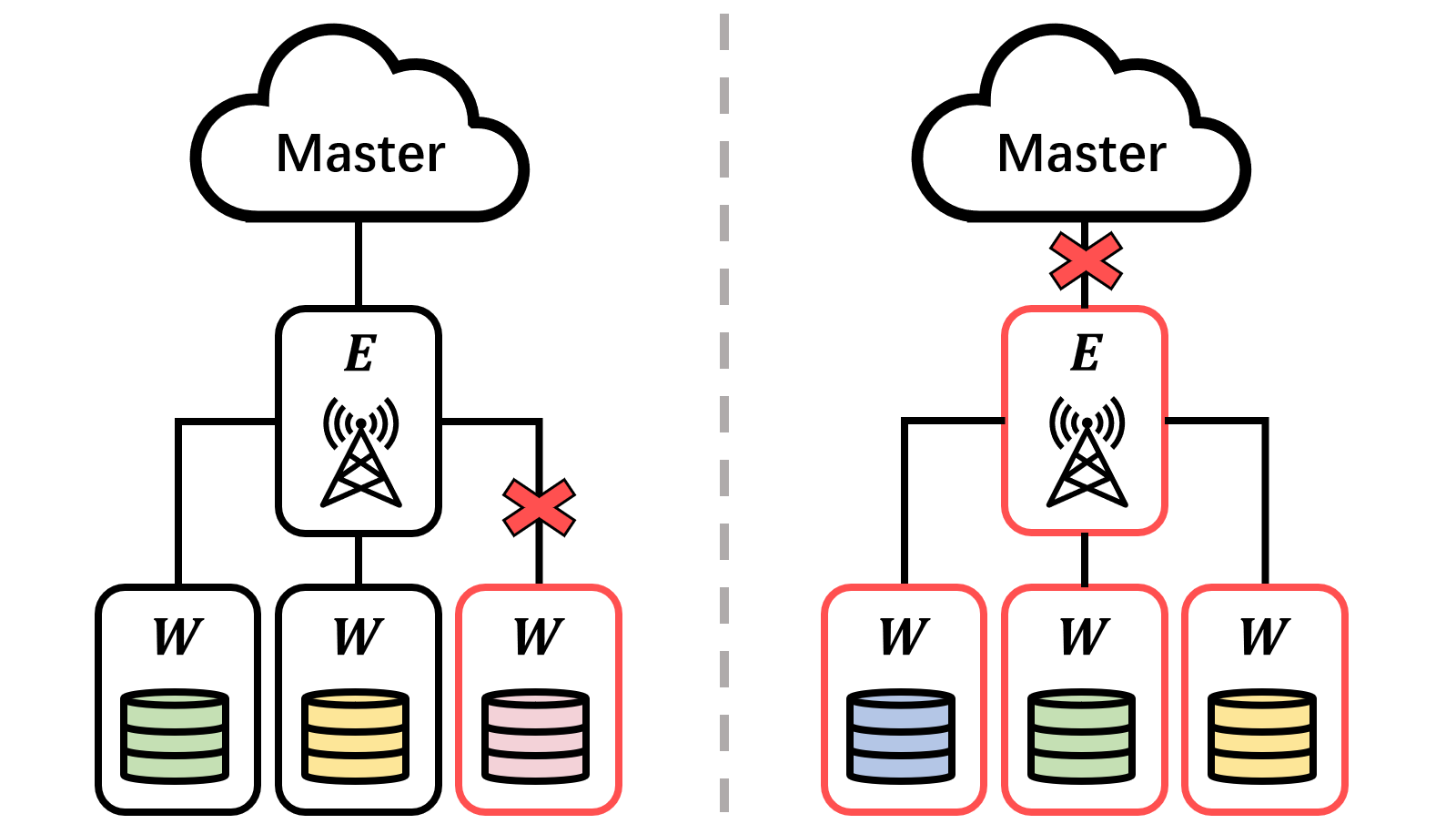}
	\caption{Illustration of the stragglers in workers and edge nodes.}
 \label{fig1-3}
\end{figure}

\subsection{Computational Trade-off}\label{set22}
Considering a hierarchical distributed learning system described in Section \ref{set21}, then we will explore the fundamental computational trade-off between the computational loads of workers and the stragglers tolerance, which also depicts the relationship between the added computing redundancy and stragglers mitigation capability.

Intuitively, to recover the full gradient while tolerating more stragglers, it is necessary to increase the computational loads of each worker. Actually, \cite{tandon2017gradient} describes the lower bound of computational loads to tolerate any $s_{w}$ straggling workers in conventional worker-master topology, i.e., single-layer architecture with only straggling workers:
\begin{equation}\label{3}
    \frac{D}{K}\geq\frac{s_{w}+1}{m}.
\end{equation}
This inequality also depicts the relationship between the added computing redundancy and stragglers tolerance. And we briefly introduce the process of obtaining the above inequality. To begin with, we consider a sub-dataset $\mathcal{D}_{i}$, and its gradient result can be recovered from any $m-s_{w}$ out of $m$ workers, only if this sub-dataset is computed by at least $s_{w}+1$ ones. For total $K$ sub-datasets, there will be not less than $K(s_{w}+1)$ computational loads to be allocated. With $D$ allocated loads each worker, hence we obtain inequality (\ref{3}) to ensure that any sub-dataset's result can be recovered. And several gradient coding schemes have been designed to achieve this trade-off, which implies that the computational loads of each worker can reach the lower bound above.

In hierarchical distributed learning systems, both stragglers of workers and edge nodes should be considered in the computational trade-off. And accordingly, we derive the new result on the computational trade-off as follows:
\begin{theorem}
Given a hierarchical distributed learning system with $n$ edge nodes, the $i$-th edge node $E_{i}$ connects to $m_{i}$ workers and $m=\min_{i} m_{i}$. Every worker will train $D$ of $K$ disjoint sub-datasets. To tolerate any $s_{e}\in[0\colon n)$ edge stragglers and any $s_{w}\in[0\colon m)$ stragglers of workers in each edge node, computational loads of each worker should satisfy
\begin{equation}\label{p1}
\frac{D}{K}\geq\frac{(s_{e}+1)(s_{w}+1)}{\sum^{n}_{i=1}m_{i}}.
\end{equation}
\end{theorem}
\begin{proof}
We consider the edge stragglers first, the gradient results of sub-datasets $\mathcal{D}_{i}$ is able to be recovered from any $f_{e}$ edge nodes only if it can be received from at least $s_{e}+1$ edge nodes. Hence, for total $K$ sub-datasets, we need $K(s_{e}+1)$ computational loads at least, and we have
\begin{equation}
    \sum^{n}_{i=1}m_{i}D\geq K(s_{e}+1).
\end{equation}
To balance loads of each worker, there should be at least $K(s_{e}+1)\frac{m_{i}}{\sum^{n}_{i=1}m_{i}}$ computational loads attained by all the workers of edge node $E_{i}$. Considering any $s_{w}$ stragglers of workers in each edge node, similarly, we have
\begin{equation}
    m_{i}D\geq K(s_{e}+1)(s_{w}+1)\frac{m_{i}}{\sum^{n}_{i=1}m_{i}}.
\end{equation}
Equivalently, we have proved
\begin{equation}
    \frac{D}{K}\geq \frac{(s_{e}+1)(s_{w}+1)}{\sum^{n}_{i=1}m_{i}}.
\end{equation}
\end{proof}

Notice that the special case $n=1$ in Theorem 1, i.e. there is only one edge node, is the same as the case considered in conventional workers-master architecture. According to Theorem 1, workers attain more computational loads because of the edge stragglers, which is aforementioned as the more severe stragglers effect. In the rest of this paper, we assume that $\frac{\sum_{i\in \mathcal{F},|\mathcal{F}|=f_{e}}m_{i}(s_{e}+1)}{\sum_{i}m_{i}}\geq 1$, which ensures that there are sufficient workers to complete the training tasks, and guarantees the existence of gradient coding schemes achieving the optimal computational trade-off. And we provide the following example to illustrate such a trade-off.
\begin{figure}[t]
	\centering
	\includegraphics[width=0.5\textwidth]{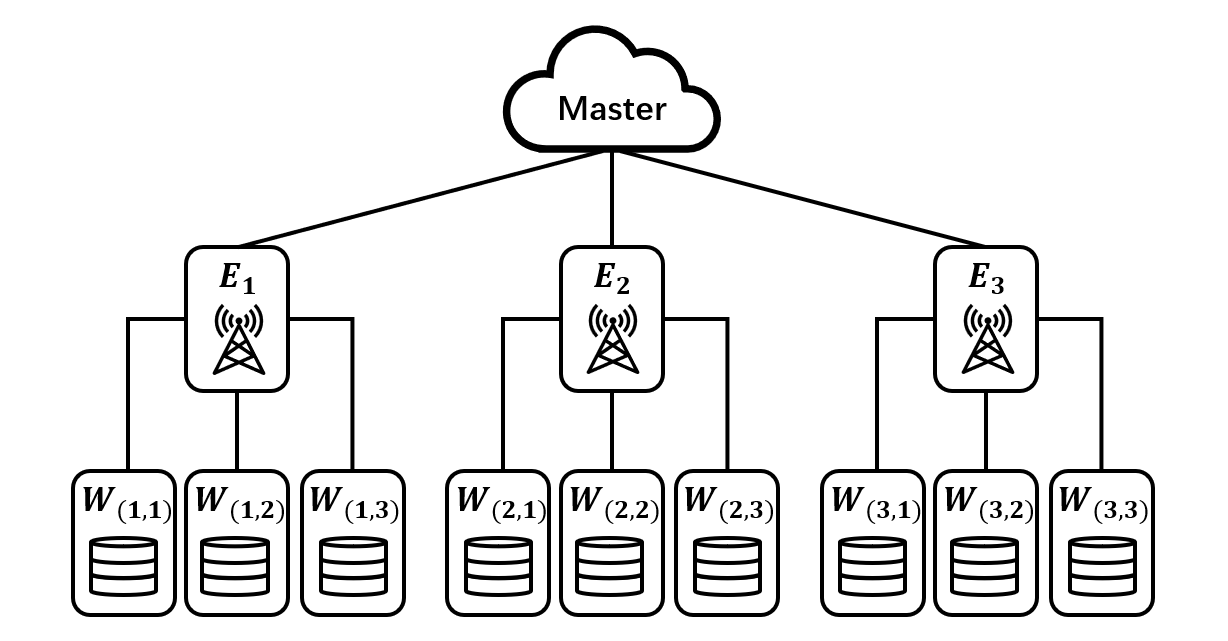}
	\caption{A hierarchical distributed system with $3$ edge nodes $E_{1},E_{2},E_{3}$, each of them connects to $3$ workers separately and interacts to a same master. All of the data will be divided into 9 disjoint sub-datasets.}
\label{exp}
\end{figure}

\textbf{Example 1}: In the hierarchical distributed learning system depicted in Fig. \ref{exp}, we first intend to tolerate arbitrary one straggling edge node, which means the results of a fixed sub-datasets can be received from at least 2 edge nodes. Consequently, the total number of sub-datastes to be allocated will be not less than 18, with 6 sub-datasets to each one. Then, we continue to tolerate arbitrary one straggling worker in each edge node. Similarly, each edge node needs to receive the results of a sub-dataset from at least 2 workers of it, resulting in a total allocation of 12 sub-datasets. Hence, each worker obtains 4 sub-datasets, i.e., its computational loads.

Then, we show that conventional gradient coding schemes cannot achieve the optimal trade-off.
\begin{corollary}\label{coro1}
Given a hierarchical distributed learning system with $n$ edge nodes, and the $i$-th edge node $E_{i}$ connects to $m_{i}$ workers, to tolerate any $s_{e}\in[0\colon n)$ edge stragglers and any $s_{w}\in[0\colon m)$ stragglers of workers in each edge node, conventional gradient coding schemes designed between workers and the master, which only consider single-layer coding, can not achieve the optimal trade-off.
\end{corollary}
\begin{proof}
For single-layer coding schemes designed between workers and the master, the decoding in master needs to deal with both straggling edge nodes and workers at the same time. As we mentioned in Section I, all the workers connected to an edge straggler will also be implicated to become stragglers. Hence, considering $s_{e}$ edge stragglers and $s_{w}$ straggling workers of each edge nodes, the maximum number of stragglers handled by the master is
\begin{equation}
s_{max}=\max_{|\mathcal{S}_{e}|=s_{e},\mathcal{S}_{e}\subset\mathcal{E}}\sum_{E_{i}\in\mathcal{S}_{e}}m_{i}+(n-s_{e})s_{w}.
\end{equation}
According to the computational loads of conventional coding design, computational loads of each worker will be
\begin{equation}\label{rmk1}
    \frac{D_{con}}{K}=\frac{\max_{|\mathcal{S}_{e}|=s_{e},\mathcal{S}_{e}\subset\mathcal{E}}\sum_{E_{i}\in\mathcal{S}_{e}}m_{i}+(n-s_{e})s_{w}+1}{\sum_{i=1}^{n}m_{i}}.
\end{equation}
For convenience, the edge stragglers set and the whole edge nodes set are denoted by $\mathcal{S}_{e}$ and $\mathcal{E}$ above. And we can further lower bound (\ref{rmk1}) as
\begin{equation}
\frac{D_{con}}{K}\geq \frac{s_{e}\min_{E_{i}\in\mathcal{E}}m_{i}+ns_{w}-s_{e}s_{w}+1}{\sum_{i=1}^{n}m_{i}}.
\end{equation}

Notice that $n\geq1+s_{e}$ and $\min_{E_{i}\in\mathcal{E}}m_{i}\geq1+s_{w}$, and these two inequalities can not be tight simultaneously, otherwise there may be only one worker compute the total gradients which is actually a centralized computation. Hence, we have $s_{e}\min_{E_{i}\in\mathcal{E}}m_{i}+ns_{w}-s_{e}s_{w}+1\textgreater (s_{e}+1)(s_{w}+1)$. Equivalently, we have proved
\begin{equation}
     \frac{D_{con}}{K}\textgreater\frac{(s_{e}+1)(s_{w}+1)}{\sum^{n}_{i=1}m_{i}}.
\end{equation}
The inequality above indicates that conventional gradient coding schemes designed between workers and the master cannot achieve the optimal trade-off.
\end{proof}

According to Corollary \ref{coro1}, conventional gradient coding schemes lead to more computational loads than the optimal trade-off due to the existence of edge stragglers. Moreover, conventional gradient coding schemes designed between workers and the master overlook the decoding capabilities of edge nodes. As a result, the partial gradient results cannot be decoded and pre-aggregated in edge nodes without accessing to the global stragglers information. Consequently, this limitation may also lead to higher communication cost. And we conclude this section by extending the above trade-off to the multi-layer distributed learning systems.
\begin{corollary}
Given a multi-layer distributed learning system with a master interacting with nodes of $L$ layers, totally $W$ workers are in the $L$-th layer. Each node in the $(i-1)$-th layer connects to disjoint $n_{i}$ layer nodes in the $i$-th layer. To tolerate $s_{i}$ stragglers of each node in the $(i-1)$-th layer (suppose the master is in the $0$-th layer), computational loads of each worker should satisfy
\begin{equation}
    \frac{D}{K}\geq\frac{\prod\limits^{L}_{i=1}(s_{i}+1)}{W}.
\end{equation}
\end{corollary}
\section{Hierarchical Gradient Coding Scheme Design}\label{set3}
In this section, we propose a hierarchical gradient coding framework to provide better stragglers mitigation for more severe stragglers effect in hierarchical architecture, which achieves the optimal trade-off between computational loads and stragglers tolerance as shown in Section \ref{set22}. Our proposed hierarchical coding scheme encodes the partial gradient results with two layers of codes during the encoding phase and utilizes the decoding capabilities of both the master and edge nodes during the decoding phase. Compared to the conventional gradient coding schemes, the proposed scheme achieves the same stragglers tolerance while allocating less computational loads. To begin with, we denote the full gradient as $g$ and the partial gradient of sub-dataset $\mathcal{D}_{i}$ as $g_{i}$ for convenience. According to equation (\ref{2}), we have
\begin{equation}
  g=\sum^{K}_{i=1}g_{i}.
\end{equation}

The entire hierarchical gradient coding scheme can be divided into two phases: encoding phase and decoding phase. And we introduce the encoding phase first.

\subsection{Encoding Phase}
During the encoding phase, to utilize the decoding capabilities of both the master and edge nodes, we correspondingly design two layers of encoding for the partial gradient results from workers: encoding between the master and edge nodes, and encoding between each edge node and its workers. Starting with the edge stragglers only, to tolerate any $s_{e}$ edge stragglers, which means the master is able to recover the full gradient from any $f_{e}$ fastest edge nodes, we need an encoding matrix $\boldsymbol{B}=[\boldsymbol{b}_{1}^{\top},\boldsymbol{b}_{2}^{\top},...,\boldsymbol{b}_{n}^{\top}]^{\top} \in\mathbb{R}^{n\times K}$ for the first layer encoding, which satisfies the following condition.
\begin{condition}
For any subset $I\subset [n]$, $|I|=f_{e}$,
\begin{equation}\label{con1}
    \boldsymbol{1}_{1\times K}\in\boldsymbol{\mathrm{span}}(\{\boldsymbol{b}_{i}|i\in I\}),
\end{equation}
where $\boldsymbol{\mathrm{span}}(\cdot)$ is the span of vectors.
\label{cond1}
\end{condition}

Such an encoding matrix satisfying (\ref{con1}) can be constructed based on conventional heterogeneous gradient coding schemes (e.g.,\cite{wang2019heterogeneity}), because different numbers of sub-datasets may be attained by different edge nodes. Denoting the number of sub-datasets attained by edge node $E_{i}$ as $n_{i}$, we calculate it as
\begin{equation}\label{s31}
    n_{i}=K(s_{e}+1)\frac{m_{i}}{\sum_{i=1}^{n}m_{i}}.
\end{equation}
And the set of these $n_{i}$ sub-datasets $\mathcal{D}^{i}$ is defined as
\begin{equation}
        \mathcal{D}^{i}=\{\mathcal{D}_{(\sum_{j=1}^{i-1}n_{j}+1)\bmod K}, \cdots ,\mathcal{D}_{(\sum_{j=1}^{i-1}n_{j}+n_{i})\bmod K}\}.
\end{equation}

Then, we construct the first layer encoding matrix $\boldsymbol{B}$ to tolerate $s_{e}$ edge stragglers. By denoting the support structure of $\boldsymbol{B}$ as $\mathrm{supp}(\boldsymbol{B})$, $\mathrm{supp}(\boldsymbol{B})$ satisfies that $\mathrm{supp}(b_{ij})=\star$ if $\mathcal{D}_{j}\in\mathcal{D}^{i}$ else $\mathrm{supp}(b_{ij})=0$, where $\star$ indicates a non-zero entry. And the linear combination of the partial gradients will be sent from the edge node $E_{i}$ is given by
\begin{equation}
G_{i}=\boldsymbol{b}_{i}\left[g_{1},g_{2},...,g_{K}\right]^{\top}=\sum^{K}_{j=1}b_{ij}g_{j}.
\end{equation}
Specifically, $G_{i}$ is the result combined by edge node $E_{i}$ from its workers. Considering $s_{w}$ stragglers of workers in edge node $E_{i}$ further, we need another encoding matrix $\boldsymbol{D}^{i}\in\mathbb{R}^{m_{i}\times K}$ for the second layer encoding. We calculate the number of sub-datasets attained by each worker of edge node $E_{i}$ as
\begin{equation}\label{s32}
    D=n_{i}\frac{s_{w}+1}{m_{i}},
\end{equation}
and the corresponding sub-datasets attained by $W_{(i,j)}$ can be defined as
\begin{equation}
        \mathcal{D}^{(i,j)}=\{\mathcal{D}^{i}_{((j-1)D+1)\bmod n_{i}},\cdots,\mathcal{D}^{i}_{((j-1)D+D)\bmod n_{i}}\}.
\end{equation}

Encoding matrix $\boldsymbol{D}^{i}$ is constructed from a matrix $\bar{\boldsymbol{D}}^{i}\in\mathbb{R}^{m_{i}\times n_{i}}$, which is constructed for the $n_{i}$ sub-datasets attained by edge node $E_{i}$ to tolerate $s_{w}$ stragglers based on conventional gradient coding schemes (e.g., \cite{tandon2017gradient,raviv2018gradient}). Similarly to encoding matrix $\boldsymbol{B}$, $\bar{\boldsymbol{D}}^{i}$ satisfies the following condition:
\begin{condition}
For any subset $I\subset [m_{i}]$, $|I|=f_{w}^{i}$,
\begin{equation}\label{con2}
    \boldsymbol{1}_{1\times n_{i}}\in\boldsymbol{\mathrm{span}}(\{\bar{\boldsymbol{D}}^{i}_{j}|j\in I\}),
\end{equation}
\label{con2}
where $\bar{\boldsymbol{D}}^{i}_{j}$ is the $j$-th row of $\bar{\boldsymbol{D}}^{i}$.
\end{condition}

And $\boldsymbol{D}^{i}$ is constructed as follows:
\begin{equation}
\begin{aligned}
&\boldsymbol{D}^{i}(\{j|\mathcal{D}_{j}\in\mathcal{D}^{i}\})=\bar{\boldsymbol{D}}^{i},\\
&\boldsymbol{D}^{i}(\{j|\mathcal{D}_{j}\notin\mathcal{D}^{i}\})=\boldsymbol{0}_{m_{i}\times (K-n_{i})},
\end{aligned}
\end{equation}
where $\boldsymbol{D}^{i}(\cdot)$ indicates the sub-matrix composed of given columns of $\boldsymbol{D}^{i}$. The support structure of $\boldsymbol{D}^{i}$, $\mathrm{supp}(\boldsymbol{D}^{i})$, satisfies that $\mathrm{supp}(d^{i}_{jk})=\star$ if $\mathcal{D}_{k}\in\mathcal{D}^{(i,j)}$ else $\mathrm{supp}(d^{i}_{jk})=0$. Denoting the $j$-th row of $\boldsymbol{D}^{i}$ as $\boldsymbol{d}^{i}_{j}$, worker $W_{(i,j)}$ encodes the partial gradients as
\begin{equation}
    G_{ij}=\boldsymbol{d}^{i}_{j}\boldsymbol{\mathrm{diag}}(g_{1},g_{2},...,g_{K})\boldsymbol{b}_{i}^\top.
\end{equation}

It is worth mentioning that, according to equations (\ref{s31}) and (\ref{s32}), one can verify that the encoding phase above satisfies
\begin{equation}
    D=K\frac{(s_{e}+1)(s_{w}+1)}{\sum^{n}_{i=1}m_{i}},
\end{equation}
which means it can make the inequality (\ref{p1}) in Theorem 1 tight, achieving the optimal trade-off between the computational loads of workers and the stragglers tolerance in a hierarchical distributed learning system. Hence, according to the Corollary 1, our proposed hierarchical gradient coding framework provides better stragglers mitigation for both stragglers in edge nodes and workers compared to the conventional gradient coding schemes. Next, we will introduce our decoding phase, which utilizes decoding capabilities from both the master and the edge nodes, hence it is able to decode the two-layer encoded results and obtain the final gradients.

\subsection{Decoding Phase}
The decoding phase can also be divided into two layers: decoding in each edge node and decoding in the master. Let $\boldsymbol{D}^{i}_{\mathcal{F}_{i}}$ and $\bar{\boldsymbol{D}}^{i}_{\mathcal{F}_{i}}$ indicate the sub-matrices of $\boldsymbol{D}^{i}$ and $\bar{\boldsymbol{D}}^{i}$ respectively, with rows indexed by $\mathcal{F}_{i}$, where $\mathcal{F}_{i}\subset[m_{i}]$ represents the set of fastest workers in edge node $E_{i}$ and $|\mathcal{F}_{i}|=f_{w}^{i}$. And the received results of the edge node $E_{i}$ can be expressed as $\boldsymbol{D}^{i}_{\mathcal{F}_{i}}\boldsymbol{\mathrm{diag}}(g_{1},g_{2},...,g_{K})\boldsymbol{b}_{i}^\top$.

According to Condition \ref{con2}, there exists a row vector $\boldsymbol{c}^{i}_{\mathcal{F}_{i}}$ that satisfies
\begin{equation}
\boldsymbol{c}^{i}_{\mathcal{F}_{i}}\bar{\boldsymbol{D}}^{i}_{\mathcal{F}_{i}}=\boldsymbol{1}_{1\times n_{i}}.
\end{equation}
Recall the construction of encoding matrix $\boldsymbol{D}^{i}$, the result of $\boldsymbol{c}^{i}_{\mathcal{F}_{i}}\boldsymbol{D}^{i}_{\mathcal{F}_{i}}$ is a row vector with $K$ elements where the $k$-th element of it will be 1 if $\mathcal{D}_{j}\in\mathcal{D}^{i}$ and be 0 else. Hence, during the first layer of decoding, edge node $E_{i}$ can decode and obtain the result as
\begin{equation}  \boldsymbol{c}^{i}_{\mathcal{F}_{i}}\boldsymbol{D}^{i}_{\mathcal{F}_{i}}\boldsymbol{\mathrm{diag}}(g_{1},g_{2},...,g_{K})\boldsymbol{b}_{i}^\top=\sum^{K}_{j=1}\boldsymbol{b}_{ij}g_{j}=G_{i}.
\end{equation}

Similarly, let $\boldsymbol{B}_{\mathcal{F}}$ indicate the sub-matrix of $\boldsymbol{B}$ with rows indexed by $\mathcal{F}$, where $\mathcal{F}\subset[n]$ represents the set of fastest edge nodes and $|\mathcal{F}|=f_{e}$. And the received results of the master can be expressed as $\boldsymbol{B}_{\mathcal{F}}[g_{1},g_{2},...,g_{K}]^{\top}$. According to Condition \ref{cond1}, there exists a row vector $\boldsymbol{a}_{\mathcal{F}}$ that satisfies
\begin{equation}
    \boldsymbol{a}_{\mathcal{F}}\boldsymbol{B}_{\mathcal{F}}=\boldsymbol{1}_{1\times K}.
\end{equation}
During the second layer of decoding, the master can decode to obtain the full gradient as
\begin{equation}
    \boldsymbol{a}_{\mathcal{F}}\boldsymbol{B}_{\mathcal{F}}
    \left[g_{1},g_{2},...,g_{K}\right]^{\top}=\sum^{K}_{i=1}g_{i}=g.
\end{equation}

\begin{algorithm}[t]
    \caption{Hierarchical Gradient Coding}
    \begin{algorithmic}[1]
    \REQUIRE $\mathcal{D},n,m_{i},s_{e},s_{w}$
        \STATE\textbf{Encoding Phase}
        \STATE Construct the first layer encoding matrix $\boldsymbol{B}\in\mathbb{R}^{n\times K}$.
        \FOR{each $i\in[n]$}
            \STATE Edge node $E_{i}$ attains sub-dateset $\mathcal{D}^{i}$.
            \STATE Construct the matrix $\bar{\boldsymbol{D}}^{i}\in\mathbb{R}^{m_{i}\times n_{i}}$.
            \STATE Construct the second layer encoding matrix $\boldsymbol{D}^{i}\in\mathbb{R}^{m_{i}\times K}$ based on $\bar{\boldsymbol{D}}^{i}$.
            \FOR{each $j\in[m_{i}]$}
                 \STATE Worker $W_{(i,j)}$ attains sub-dataset $\mathcal{D}^{(i,j)}$.
                 \STATE Worker $W_{(i,j)}$ encodes the partial gradients as $$G_{ij}=\boldsymbol{d}^{i}_{j}\boldsymbol{\mathrm{diag}}(g_{1},g_{2},...,g_{K})\boldsymbol{b}_{i}^\top.$$
            \ENDFOR
        \ENDFOR
        \STATE \textbf{Decoding Phase}
        \FOR{each $i\in[n]$}
            \STATE  Edge node $E_{i}$ chooses the corresponding decoding vector $\boldsymbol{c}^{i}_{\mathcal{F}_{i}}$ to implement the first layer decoding as $$G_{i}=\boldsymbol{c}^{i}_{\mathcal{F}_{i}}\boldsymbol{D}^{i}_{\mathcal{F}_{i}}\boldsymbol{\mathrm{diag}}(g_{1},g_{2},...,g_{K})\boldsymbol{b}_{i}^\top.$$
        \ENDFOR
        \STATE The master chooses the corresponding decoding vector $\boldsymbol{a}_{\mathcal{F}}$ to implement the second layer decoding as $$g=\boldsymbol{a}_{\mathcal{F}}\boldsymbol{B}_{\mathcal{F}}[g_{1},g_{2},...,g_{K}]^{\top}.$$
    \end{algorithmic}
    \label{alg1}
\end{algorithm}
We then summarize the whole hierarchical gradient coding scheme in Algorithm \ref{alg1}. Additionally, accelerative single-layer gradient coding techniques like utilizing partial computing results\cite{DSW8755563} can also be combined in coding between workers and edge nodes. In the description above, we still consider full computing results from workers to introduce a basic hierarchical coding framework. And we also demonstrate our hierarchical gradient coding scheme with an example below.

\begin{figure}[t]
	\centering
	\includegraphics[width=0.5\textwidth]{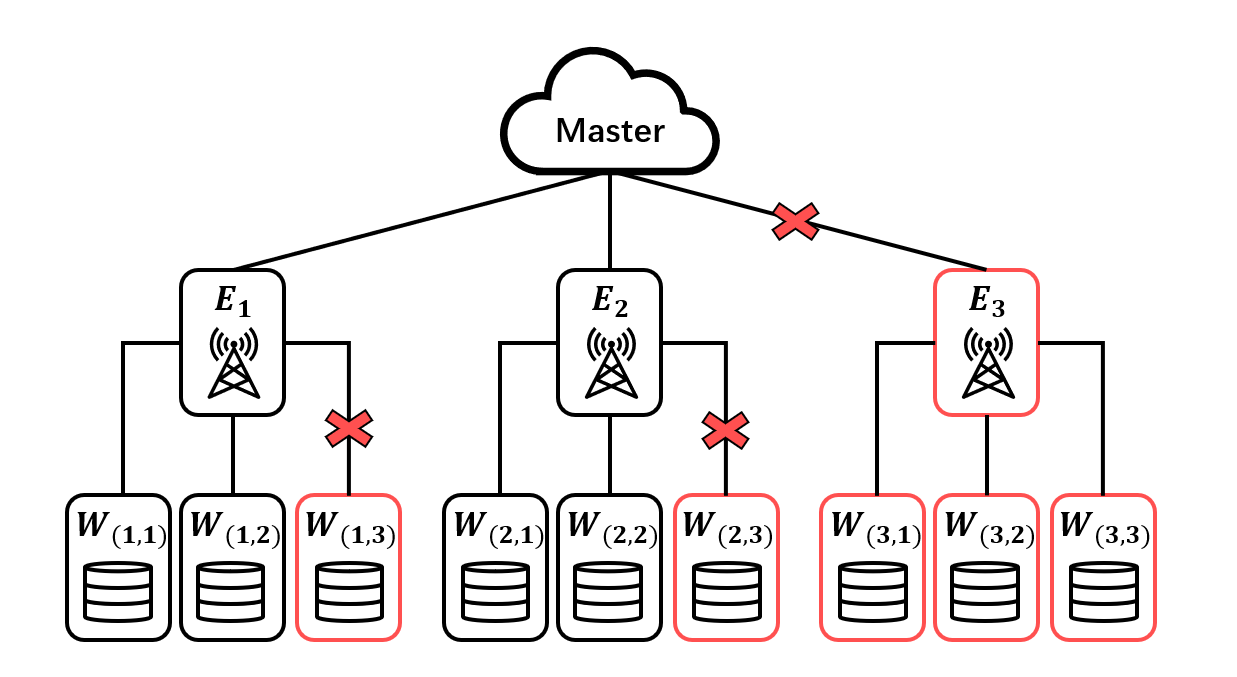}
	\caption{A hierarchical distributed system with straggling edge node $E_{3}$, worker $W_{(1,3)}$, and worker $W_{(2,3)}$.}
 \label{fig2}
\end{figure}

\textbf{Example 1 (continued)}: Considering the hierarchical distributed learning system illustrated in Fig. \ref{fig2}, in edge node $E_{1}$, worker $W_{(1,3)}$ is the straggler, then workers $W_{(1,1)}$ and $W_{(1,2)}$ encode the partial gradients as
\begin{align*}
    G_{11}&=\frac{1}{4}g_{1}+\frac{1}{4}g_{2}+\frac{1}{2}g_{3}+g_{4},\\ G_{12}&=\frac{1}{4}g_{1}+\frac{1}{4}g_{2}+g_{5}+g_{6}.
\end{align*}
Hence, edge node $E_{1}$ can decode
$$G_{1}=G_{11}+G_{12}=\frac{1}{2}g_{1}+\frac{1}{2}g_{2}+\frac{1}{2}g_{3}+g_{4}+g_{5}+g_{6}.$$
Similarly, for edge node $E_{2}$, it can decode
$$G_{2}=G_{21}+G_{22}=\frac{1}{2}g_{1}+\frac{1}{2}g_{2}+\frac{1}{2}g_{3}+g_{7}+g_{8}+g_{9},$$ from $G_{21}=\frac{1}{4}g_{1}+\frac{1}{4}g_{2}+\frac{1}{2}g_{3}+g_{7}$ and $G_{22}=\frac{1}{4}g_{1}+\frac{1}{4}g_{2}+g_{8}+g_{9}$, without receiving worker $W_{(2,3)}$'s encoded result. For the master, while $E_{3}$ is the edge straggler, it can decode the full gradient $g=G_{1}+G_{2}$ at the end.

\section{Runtime Analysis and Design Optimization}\label{set5}
In this section, we first model the total runtime of each iteration in the proposed hierarchical coded distributed learning system. Subsequently, runtime performance analysis based on our modeling is conducted for fully homogeneous hierarchical coded distributed learning systems with two special cases. After that, to further improve the performance of our proposed framework in heterogeneous scenarios, we formulate a jointly node and coding scheme selection problem (JNCSS) and develop an optimization algorithm for heterogeneous hierarchical coded distributed learning systems, aiming to minimize the expected execution time for each iteration in the learning process.
\subsection{Runtime Model}\label{set51}
To begin with, in this subsection, we formulate the runtime model of each iteration in the proposed hierarchical coded distributed learning system.

For worker $W_{(i,j)}$, we assume a shifted exponential model for its gradient computation. More specifically, computation time of worker $W_{(i,j)}$ can be expressed as
\begin{equation}
    T^{(i,j)}_{cmp}= T^{(i,j)}_{cmp,1}+ T^{(i,j)}_{cmp,2}.
\end{equation}
$T^{(i,j)}_{cmp,1}$ above represents the deterministic computation time to process the data, which is proportional to the number of sub-datasets. Hence, we have $T^{(i,j)}_{cmp,1}=c_{(i,j)}D$. And $T^{(i,j)}_{cmp,2}$ represents the stochastic computation time due to random memory access during read/write cycles associated with Multiply-Accumulate (MAC) operations, which follows an exponential distribution, i.e., $f_{T^{(i,j)}_{cmp,2}}(t)=\gamma_{(i,j)}e^{-\gamma_{(i,j)}t},t\geq 0$. And the average computation time of worker $W_{(i,j)}$ is given by $\mathbb{E}(T^{(i,j)}_{cmp})=c_{(i,j)}D+\frac{1}{\gamma_{(i,j)}}$.

In addition to computing the gradient results, worker $W_{(i,j)}$ communicates with its edge node $E_{i}$ for data downloading and uploading during each iteration. For the communication time between edge node $E_{i}$ and worker $W_{(i,j)}$, we denote them as $T^{(i,j)}_{com,d}$ and $T^{(i,j)}_{com,u}$, and we have
\begin{equation}
T^{(i,j)}_{com,d}=N^{d}_{(i,j)}\tau_{(i,j)},T^{(i,j)}_{com,u}=N^{u}_{(i,j)}\tau_{(i,j)},
\end{equation}
where $\tau_{(i,j)}$ is the deterministic time of worker $W_{(i,j)}$ to download or upload the gradient results with the assumption that downloading and uploading delays are reciprocal. $N^{d}_{(i,j)}$ and $N^{u}_{(i,j)}$ denote the number of transmissions for successful downloading and uploading. Considering that the communication link is unreliable and may fail especially in wireless scenarios, workers or edge nodes retransmit until the successful transmission occurs. We assume a geometric distribution for the number of transmissions, i.e., $Pr(N^{d}_{(i,j)}=x)=Pr(N^{u}_{(i,j)}=x)=p_{(i,j)}^{x-1}(1-p_{(i,j)}),x\in\mathbb{Z}^{+}$. And the average downloading and uploading time between worker $W_{(i,j)}$ and edge node $E_{i}$ are given by $\mathbb{E}(T^{(i,j)}_{com,d})=\mathbb{E}(T^{(i,j)}_{com,u})=\frac{\tau_{(i,j)}}{1-p_{(i,j)}}$.

Similarly, edge node $E_{i}$ also communicates with the central master for data downloading and uploading during each iteration. Denoting the communication time as $T^{i}_{com,d}$ and $T^{i}_{com,u}$, we have
\begin{equation}
    T^{i}_{com,d}=N^{d}_{i}\tau_{i},T^{i}_{com,u}=N^{u}_{i}\tau_{i}.
\end{equation}
Here, $\tau_{i}$ is the deterministic time of edge node $E_{i}$ to download or upload the gradient results. $N^{d}_{i}$ and $N^{u}_{i}$ denote the number of transmissions for successful downloading and uploading, which also follow the geometric distribution, i.e., $Pr(N^{d}_{i}=x)=Pr(N^{u}_{i}=x)=p_{i}^{x-1}(1-p_{i}),x\in\mathbb{Z}^{+}$. And the average downloading and uploading time between edge node $E_{i}$ and the master are given by $\mathbb{E}(T^{i}_{com,d})=\mathbb{E}(T^{i}_{com,u})=\frac{\tau_{i}}{1-p_{i}}$.

Hence, during each iteration, the total runtime of worker $W_{(i,j)}$ to download aggregated model from the central master, compute to update the latest gradient, and upload to its edge node can be expressed as
\begin{equation}
    T^{(i,j)}_{tol}=T^{i}_{com,d}+T^{(i,j)}_{com,d}+T^{(i,j)}_{cmp}+T^{(i,j)}_{com,u}.
\end{equation}

If the proposed hierarchical gradient coding scheme is considered, when $s_{w}$ stragglers of workers are tolerated in edge node $E_{i}$, the total runtime to return its results to the master is
\begin{equation}
    T^{i}_{tol}=T^{i}_{com,u}+\min_{(m_{i}-s_{w})-\mathrm{th}}\left\{T^{(i,j)}_{tol}:\forall j\in[m_{i}]\right\},
\end{equation}
where $\min\limits_{(m_{i}-s_{w})-\mathrm{th}}(\cdot)$ returns the $(m_{i}-s_{w})$-th minimum value of a set. For example, $\min\limits_{3-\mathrm{th}}\{3,4,5,6\}=5$. Moreover, when $s_{e}$ edge stragglers are tolerated, the total runtime for the hierarchical coded distributed learning system to compute and recover the full gradient is
\begin{equation}
    T_{tol}=\min_{(n-s_{e})-\mathrm{th}}\left\{T^{i}_{tol}:\forall i\in[n]\right\}.
\end{equation}
\subsection{Analysis of Homogeneous Scenarios}
In this subsection, we focus on fully homogeneous hierarchical coded distributed learning systems, where all the workers and edge nodes have similar computation and communication capacities. More precisely, we introduce two special cases aimed at minimizing the expected runtime of each iteration by choosing appropriate stragglers tolerance $(s_{e},s_{w})$.

\textbf{Case 1 (Computation-dominated)}: Considering the scenario that $p_{1}\approx 0$ and $p_{2}\approx 0$, which indicates the downloading and uploading times are fixed parameters as represented by $T^{(i,j)}_{com,d}=T^{(i,j)}_{com,u}=\tau_{1}$, $T^{i}_{com,d}=T^{i}_{com,u}=\tau_{2}$. Hence, the computation time of each worker dominates the total runtime. If any $s_{e}$ edge stragglers and any $s_{w}$ stragglers of workers in each edge node can be tolerated with the proposed hierarchical gradient coding scheme, the total runtime of each iteration will be
\begin{equation}
\begin{aligned}
T_{tol}=cK\frac{(s_{e}+1)(s_{w}+1)}{nm}+2\tau_{1}+2\tau_{2}+\max_{i\in\mathcal{F},j\in\mathcal{F}_{i}}T^{(i,j)}_{cmp,2},
\end{aligned}
\end{equation}
where statistical variable $\max_{i\in\mathcal{F},j\in\mathcal{F}_{i}}T^{(i,j)}_{cmp,2}$ is the maximum of $(n-s_{e})(m-s_{w})$ independent and identical exponential distribution variables, and it is well known and easily derived that $\mathbb{E}(\max_{i\in\mathcal{F},j\in\mathcal{F}_{i}}(T^{(i,j)}_{cmp,2}))\approx \frac{1}{\gamma}\ln (n-s_{e})(m-s_{w})$. Hence, we have
\begin{equation}\label{case1}
\begin{aligned}
    \mathbb{E}(T_{tol})\approx & cK \frac{(s_{e}+1)(s_{w}+1)}{nm}+2\tau_{1}+2\tau_{2}
    \\& +\frac{1}{\gamma}\ln (n-s_{e})(m-s_{w}).
\end{aligned}
\end{equation}

Viewing (\ref{case1}) as a binary function of $s_{e}$ and $s_{w}$, for any $s_{e}$, we observe that function (\ref{case1}) is monotonically increasing when $s_{e}\leq n-\frac{1}{\gamma}\frac{nm}{cK(s_{w}+1)}$, and monotonically decreasing when $s_{e}\textgreater  n-\frac{1}{\gamma}\frac{nm}{cK(s_{w}+1)}$. As for any $s_{w}$, function (\ref{case1}) is also monotonically increasing when $s_{w}\leq m-\frac{1}{\gamma}\frac{nm}{cK(s_{e}+1)}$, and monotonically decreasing when $s_{w}\textgreater  m-\frac{1}{\gamma}\frac{nm}{cK(s_{e}+1)}$. Notice that $s_{e}\in[0\colon n)$ and $s_{w}\in[0\colon m)$ in the proposed coding scheme, the expected runtime takes the minimum when stragglers tolerance $(s_{e},s_{w})$ is at the four endpoints of its domain, i.e. $(0,0)$, $(n-1,0)$, $(0,m-1)$ and $(n-1,m-1)$.

To simplify notations, we denote $C_{1}=\min\{cK,\frac{cK}{m}+\frac{\ln m}{\gamma},\frac{cK}{n}+\frac{\ln n}{\gamma},\frac{cK}{nm}+\frac{\ln nm}{\gamma}\}$. If $\frac{cK}{nm}+\frac{\ln nm}{\gamma}=C_{1}$, we should choose $s_{e}=0$ and $s_{w}=0$ to minimize $\mathbb{E}(T_{tol})$, and we have
\begin{equation*}
    \min{\mathbb{E}(T_{tol})}=\frac{cK}{nm}+2\tau_{1}+2\tau_{2}+\frac{\ln nm}{\gamma}.
\end{equation*}
Similarly, if $\frac{cK}{m}+\frac{\ln m}{\gamma}=C_{1}$, we should choose $s_{e}=n-1$ and $s_{w}=0$ resulting in $\min{\mathbb{E}(T_{tol})}=\frac{cK}{m}+2\tau_{1}+2\tau_{2}+\frac{\ln m}{\gamma}$. If $\frac{cK}{n}+\frac{\ln n}{\gamma}=C_{1}$, we should choose $s_{e}=0$ and $s_{w}=m-1$ resulting in $\min{\mathbb{E}(T_{tol})}=\frac{cK}{n}+2\tau_{1}+2\tau_{2}+\frac{\ln n}{\gamma}$. Otherwise, we should choose $s_{e}=n-1$ and $s_{w}=m-1$ resulting in $\min{\mathbb{E}(T_{tol})}=cK+2\tau_{1}+2\tau_{2}$.

\textbf{Case 2 (Communication-dominated)}: Considering the scenario that $p_{1}\approx 0$ and $\gamma\to +\infty$, which indicates that $T^{(i,j)}_{com,d}=T^{(i,j)}_{com,u}=\tau_{1}$, $T^{(i,j)}_{cmp,2}=0$. The communication time between each edge node and the master dominates the total runtime. Hence, the total runtime of each iteration is
\begin{equation}
    T_{tol}=cK \frac{(s_{e}+1)(s_{w}+1)}{nm}+2\tau_{1}+\tau_{2}\max_{i\in\mathcal{F}}(N^{u}_{i}+N^{d}_{i}).
\end{equation}

Different from Case 1, both $N^{u}_{i}$ and $N^{d}_{i}$ are geometric distribution variables in this case. In \cite{eisenberg2008expectation}, it has been derived that for $n$ independent and identical geometric distribution variables, the expected value of their maximum can be approximated as $\frac{1}{2}-\frac{1}{\ln p}n$. Hence, we have
\begin{equation}\label{case2}
    \mathbb{E}(T_{tol})\approx cK \frac{(s_{e}+1)(s_{w}+1)}{nm}+2\tau_{1}+\tau_{2}-\frac{2\tau_{2}}{\ln p_{2}}\ln (n-s_{e}).
\end{equation}

According to (\ref{case2}), it is apparent that we should choose $s_{w}=0$, since in such a case we should minimize the computational loads of each worker as much as possible. And the expected total runtime is given by
\begin{equation}\label{case2.1}
    \mathbb{E}(T_{tol})=cK\frac{s_{e}+1}{nm}+2\tau_{1}+\tau_{2}-\frac{2\tau_{2}}{\ln p_{2}}\ln (n-s_{e}).
\end{equation}

Viewing (\ref{case2.1}) as a function of $s_{e}$, it is easy to show that the expected runtime takes the minimum at the endpoints of the domain. When $\frac{cK}{m}\geq \frac{cK}{nm}-\frac{2\tau_{2}}{\ln p_{2}}\ln n$, we should choose $s_{e}=0$. And we have $\min\mathbb{E}(T_{tol})=\frac{cK}{nm}+2\tau_{1}+\tau_{2}-\frac{2\tau_{2}}{\ln p_{2}}\ln n$. Otherwise, we should choose $s_{e}=n-1$, which results in $\min\mathbb{E}(T_{tol})=\frac{cK}{m}+2\tau_{1}+\tau_{2}$.

\subsection{Optimization of Heterogeneous Scenarios}\label{set53}
In the previous subsection, we consider fully homogeneous scenarios with related analysis. Accordingly, we can observe that inappropriate selection of stragglers tolerance level does not decrease the expected runtime, and may even underperform the uncoded scheme.

However, conducting similar analysis for heterogeneous hierarchical coded distributed learning systems is more complex and challenging due to heterogeneity in computation and communication capabilities across edge nodes and workers, which is the second challenge in hierarchical architecture as discussed in Section \ref{set1}. Hence, to further improve the performance of our proposed framework in heterogeneous scenarios, we formulate a jointly node and coding scheme selection problem (JNCSS), which aims to minimize the total runtime $T_{tol}$ of each iteration.

We first introduce two sets of optimization variables $e_{i}$ and $w_{(i,j)}$. Variable $e_{i}$ indicates whether edge node $E_{i}$ is edge straggler ($e_{i}=0$) or not ($e_{i}=1$), and variable $w_{(i,j)}$ indicates whether worker $W_{(i.j)}$ is straggler ($w_{(i,j)}=0$) or not ($w_{(i,j)}=1$). Hence, we have
\begin{equation}\label{opt1}
    e_{i}\in\{0,1\},\ \ \forall i\in\left[n\right],
\end{equation}
\begin{equation}\label{opt2}
    w_{(i,j)}\in\{0,1\},\ \ \forall i\in\left[n\right],j\in\left[m_{i}\right].
\end{equation}
Then, the total number of edge non-stragglers is $f_{e}=n-s_{e}$, which means
\begin{equation}\label{opt3}
    \sum_{i=1}^{n}e_{i}=n-s_{e}.
\end{equation}
As for edge node $E_{i}$, all the workers connecting to it will be stragglers when $e_{i}=0$, else the number of non-stragglers will be $f_{w}^{i}=m_{i}-s_{w}$. Consequently, we have
\begin{equation}\label{opt4}
    \sum_{j=1}^{m_{i}}w_{(i,j)}=e_{i}(m_{i}-s_{w}),\ \ \forall i\in\left[n\right].
\end{equation}

Recall our the runtime model in Section \ref{set51}, we make the following approximations $T^{(i,j)}_{tol}\approx c_{(i,j)}D+\frac{1}{\gamma_{(i,j)}}+\frac{2\tau_{(i,j)}}{1-p_{(i,j)}}+\frac{\tau_{(i)}}{1-p_{(i)}}$ and $T^{(i)}_{com,u}\approx\frac{\tau_{(i)}}{1-p_{(i)}}$. And the total runtime is decided by the edge nodes and workers with the longest runtime among non-stragglers. Hence, for all the $i\in\left[n\right],j\in\left[m_{i}\right]$, we have:
\begin{equation}\label{opt5}
\begin{aligned}
    T_{tol}\geq & w_{(i,j)}\left(c_{(i,j)}D+\frac{1}{\gamma_{(i,j)}}+\frac{2\tau_{(i,j)}}{1-p_{(i,j)}}+\frac{\tau_{(i)}}{1-p_{(i)}}\right)
    \\ & +e_{i}\frac{\tau_{(i)}}{1-p_{(i)}}.
\end{aligned}
\end{equation}
computational loads $D$ above is determined by
\begin{equation}\label{opt6}
    D=K\frac{(s_{e}+1)(s_{w}+1)}{\sum^{n}_{i=1}m_{i}}.
\end{equation}
For convenience, we denote $m=\min_{i}m_{i}$, and we have the domain of the stragglers tolerance:
\begin{equation}\label{opt7}
    s_{e}\in[0\colon n),
\end{equation}
\begin{equation}\label{opt8}
    s_{w}\in[0\colon m).
\end{equation}
Hence, we formulate a jointly node and coding scheme selection problem (JNCSS) as follows:
\begin{equation}
    \begin{aligned}
        \mathcal{P}_{1}:&\min_{s_{w},s_{e},\boldsymbol{e},\boldsymbol{w}}\ \ T_{tol},\\
        &\ \ \ \ \text{s.t.}\ \ (\ref{opt1})-(\ref{opt8}).
    \end{aligned}
\end{equation}

In fact, the solution value of $\mathcal{P}_{1}$ represents the best expected performance of the proposed coding scheme when it tolerates corresponding $s_{e}$ and $s_{w}$ stragglers in edge node layer and worker layer, respectively. And such a best-case may occur in most of the time, especially in heterogeneous scenarios.

And we propose Algorithm \ref{alg2} to optimize the JNCSS problem as below. During each iteration, the algorithm starts by calculating $B_{(i,j)}=c_{(i,j)}D+\frac{1}{\gamma_{(i,j)}}+\frac{2\tau_{(i,j)}}{1-p_{(i,j)}}+\frac{\tau_{(i)}}{1-p_{(i)}}$ and sorts $B_{(i,j)}$ to choose the $(m_{i}-s_{w})$-th smallest one (Line 3-6). Similarly, it calculates $A_{i}=\frac{\tau_{(i)}}{1-p_{(i)}}$ and sorts $A_{i}+\min\limits_{(m_{i}-s_{w})-\mathrm{th}}B_{(i,j)}$ to choose the $(n-s_{e})$-th smallest one (Line 7-9). Then, the algorithm outputs $\hat{T}_{tol}=\min\limits_{s_{e},s_{w}}\hat{T}_{tol}(s_{e},s_{w})$ and $(\hat{s}_{e},\hat{s}_{w})=\mathop{\arg\max}\limits_{s_{e},s_{w}}\hat{T}_{tol}(s_{e},s_{w})$ (Line 11-12). Finally, the corresponding variables $e_{i}$ and $w_{(i,j)}$ are set to 1 (Line 13-21). Since Algorithm \ref{alg2} optimizes JNCSS with a greedy manner, we can obtain its complexity as $\mathcal{O}(n^{2}m^{3}+n^{3}m)$ in the worst case. Subsequently, we provide theoretical analysis of the results returned by Algorithm \ref{alg2}.

\begin{algorithm}[t]
    \caption{Solving JNCSS}
    \begin{algorithmic}[1]
        \REQUIRE $n,m_{i},K,c_{(i,j)},\gamma_{(i,j)},\tau_{(i,j)},p_{(i,j)},\tau_{(i)},p_{(i)}$
        \ENSURE $\hat{s}_{e},\hat{s}_{w},\hat{T}_{tol},e_{i},w_{(i,j)}$
            \FOR{$s_{e}\in[0\colon n),s_{w}\in[0\colon m)$}
            \FOR{$i\in[n]$}
            \FOR{$j\in[m_{i}]$}
            \STATE Calculate $B_{(i,j)}=c_{(i,j)}D+\frac{1}{\gamma_{(i,j)}}+\frac{2\tau_{(i,j)}}{1-p_{(i,j)}}+\frac{\tau_{(i)}}{1-p_{(i)}}$.
            \ENDFOR
            \STATE Sort $B_{(i,j)}$ and choose $\min\limits_{(m_{i}-s_{w})-\mathrm{th}}B_{(i,j)}$.
            \STATE Calculate $A_{i}=\frac{\tau_{(i)}}{1-p_{(i)}}$.
            \ENDFOR
            \STATE Sort $A_{i}+\min\limits_{(m_{i}-s_{w})-\mathrm{th}}B_{(i,j)}$ and choose $\hat{T}_{tol}(s_{e},s_{w})=\min\limits_{(n-s_{e})-\mathrm{th}}(A_{i}+\min\limits_{(m_{i}-s_{w})-\mathrm{th}}B_{(i,j)})$.
            \ENDFOR
            \STATE $\hat{T}_{tol}=\min\limits_{s_{e},s_{w}}\hat{T}_{tol}(s_{e},s_{w})$.
            \STATE $(\hat{s}_{e},\hat{s}_{w})=\mathop{\arg\max}\limits_{s_{e},s_{w}}\hat{T}_{tol}(s_{e},s_{w})$.
            \FOR{$i\in[n]$}
            \IF{$A_{i}+\min\limits_{(m_{i}-\hat{s}_{w})-\mathrm{th}}B_{(i,j)}\leq \hat{T}_{tol}$}
            \STATE $e_{i}=1$.
            \FOR{$j\in[m_{i}]$}
            \IF{$B_{(i,j)}\leq\min\limits_{(m_{i}-s_{w})-\mathrm{th}}B_{(i,j)}$}
            \STATE $w_{(i,j)}=1$.
            \ENDIF
            \ENDFOR
            \ENDIF
            \ENDFOR
    \end{algorithmic}
    \label{alg2}
\end{algorithm}

\begin{theorem}\label{theo2}
The output of Algorithm \ref{alg2}, $\hat{T}_{tol}$ is the optimal value of the JNCSS problem and the corresponding nodes selection and coding scheme achieves the optimum.
\end{theorem}

Theorem \ref{theo2} is proven through a proof by contradiction, which is provided in Appendix \ref{apc}. Subsequently, we aim to prove the upper bound of the gap between the solution value of the JNCSS problem $\mathcal{P}_{1}$ and the expected value of the runtime which is formulated in Section \ref{set51}.

\begin{theorem}\label{t2}
The expected gap between the output of Algorithm \ref{alg2}, $\hat{T}_{tol}$, and runtime $T_{tol}$ formulated in Section \ref{set51}, is upper bounded by
\begin{equation}
    \mathbb{E}\left|T_{tol}-\hat{T}_{tol}\right|\leq f(n,n-\hat{s}_{e})\Delta_{e}+\max\limits_{i}f(m_{i},m_{i}-\hat{s}_{w})\Delta^{i}_{w}.
\end{equation}
In the inequality above, to simplify notations, we denote
\begin{equation}
    \begin{aligned}
        \Delta_{e}=\bigg\{ &\sum_{i}\left[\mathbb{V}^{2}\left[T_{tol}^{i}\right]+\mathbb{E}^{2}\left[T_{tol}^{i}-\frac{\sum T_{tol}^{i}}{n}\right]\right]
        \\&-n\mathbb{V}^{2}\left[\frac{\sum T_{tol}^{i}}{n}\right]\bigg\}^{\frac{1}{2}},
        \\
        \Delta_{w}^{i}=\bigg\{ &\sum_{j}\left[\mathbb{V}^{2}\left[T_{tol}^{(i,j)}\right]+\mathbb{E}^{2}\left[T_{tol}^{(i,j)}-\frac{\sum T_{tol}^{(i,j)}}{m_{i}}\right]\right]
        \\&-n\mathbb{V}^{2}\left[\frac{\sum T_{tol}^{(i,j)}}{m_{i}}\right]\bigg\}^{\frac{1}{2}}.
    \end{aligned}
\end{equation}
Stragglers tolerance $\hat{s}_{e}$ and $\hat{s}_{w}$ are the results returned by Algorithm \ref{alg2}.
\end{theorem}

Proof of Theorem \ref{t2} is provided in Appendix \ref{apa}. According to Theorem \ref{t2}, $\Delta_{e}$ and $\Delta_{w}^{i}$ represent the heterogeneity and discreteness of computation and communication capabilities of edge nodes and workers, in a hierarchical distributed learning system. According to Theorem \ref{t2}, decreasing values of $\Delta_{e}$ and $\Delta_{w}^{i}$ result in the reduction of the expected gap between $T_{tol}$ and the result of our proposed algorithm.

\section{Numerical Results}\label{set6}
\begin{figure*}[t]
	\centering
	\subfloat[MNIST and non-IID Level \uppercase\expandafter{\romannumeral1}]{
            \label{fig3a}
		\includegraphics[width=0.325\textwidth]{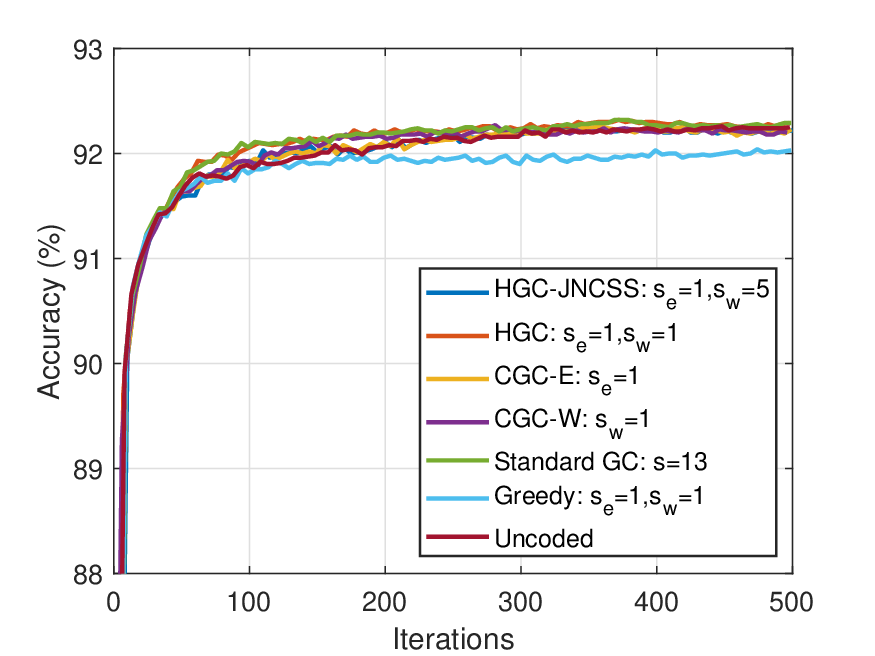}
	}
        \subfloat[MNIST and non-IID Level \uppercase\expandafter{\romannumeral2}]{
   		\label{fig3b}
            \includegraphics[width=0.325\textwidth]{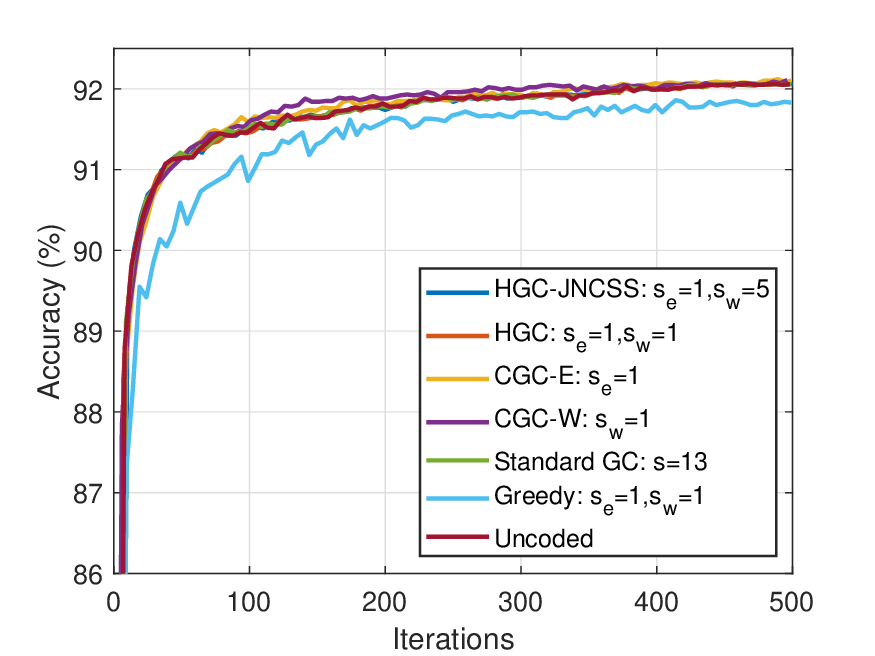}
    	}
        \subfloat[MNIST and non-IID Level \uppercase\expandafter{\romannumeral3}]{
            \label{fig3c}
	      \includegraphics[width=0.325\textwidth]{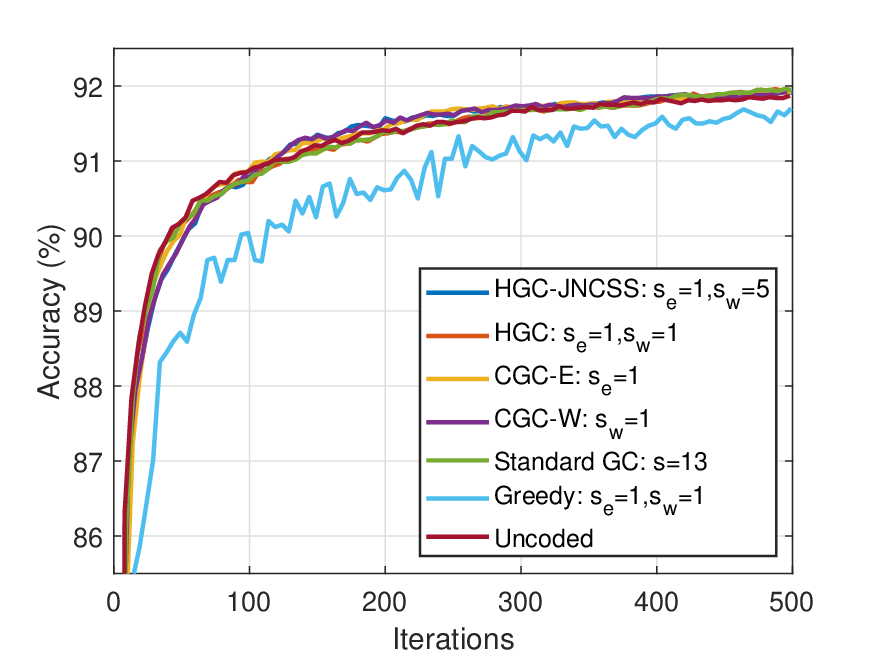}
	}
	\\
	\subfloat[CIFAR-10 and non-IID Level \uppercase\expandafter{\romannumeral1}]{
            \label{fig3d}
		\includegraphics[width=0.325\textwidth]{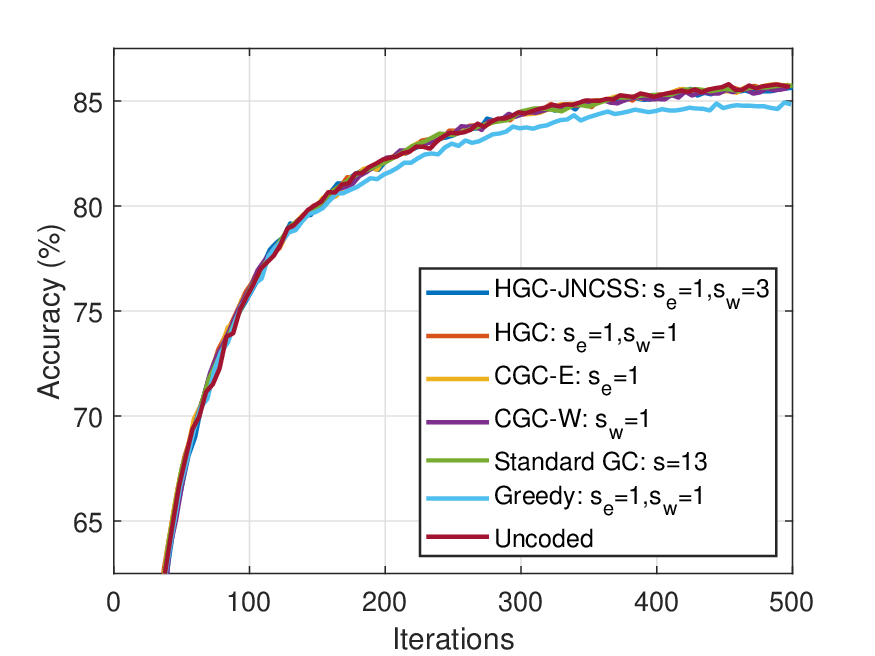}
	}
        \subfloat[CIFAR-10 and non-IID Level \uppercase\expandafter{\romannumeral2}]{
            \label{fig3e}
		\includegraphics[width=0.325\textwidth]{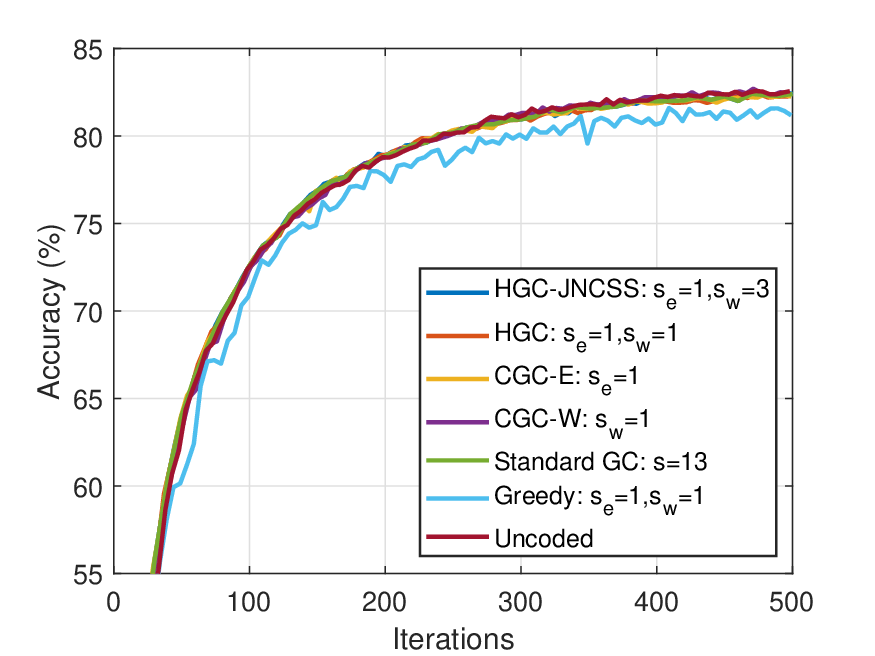}
    	}
        \subfloat[CIFAR-10 and non-IID Level \uppercase\expandafter{\romannumeral3}]{
            \label{fig3f}
		\includegraphics[width=0.325\textwidth]{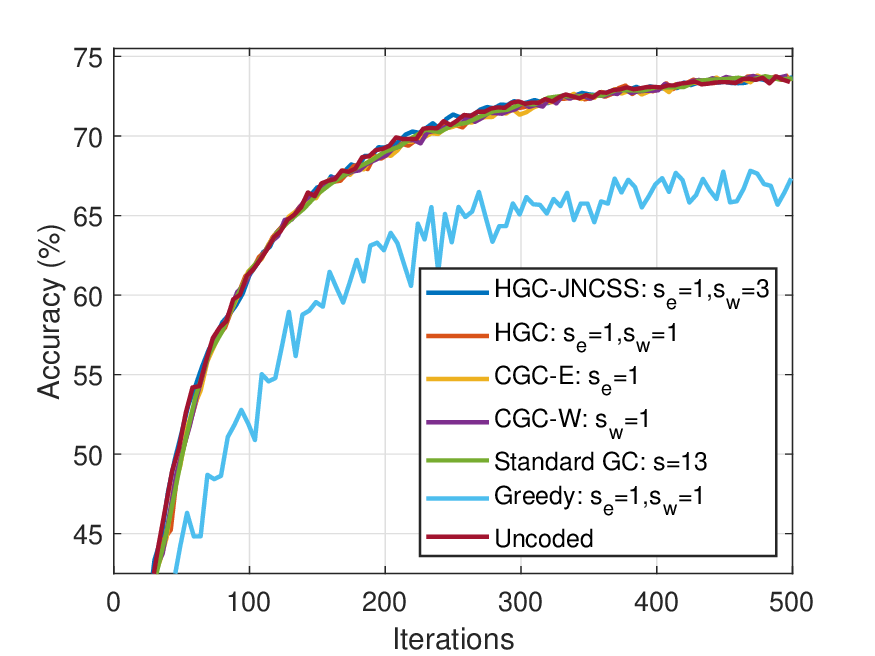}
    	}
	\caption{Test accuracy curves with respect to training iterations for different datasets and data non-IID levels.}
	\label{fig3}
\end{figure*}

\begin{figure*}[t]
	\centering
	\subfloat[MNIST and non-IID Level \uppercase\expandafter{\romannumeral1}]{
            \label{fig4a}
		\includegraphics[width=0.325\textwidth]{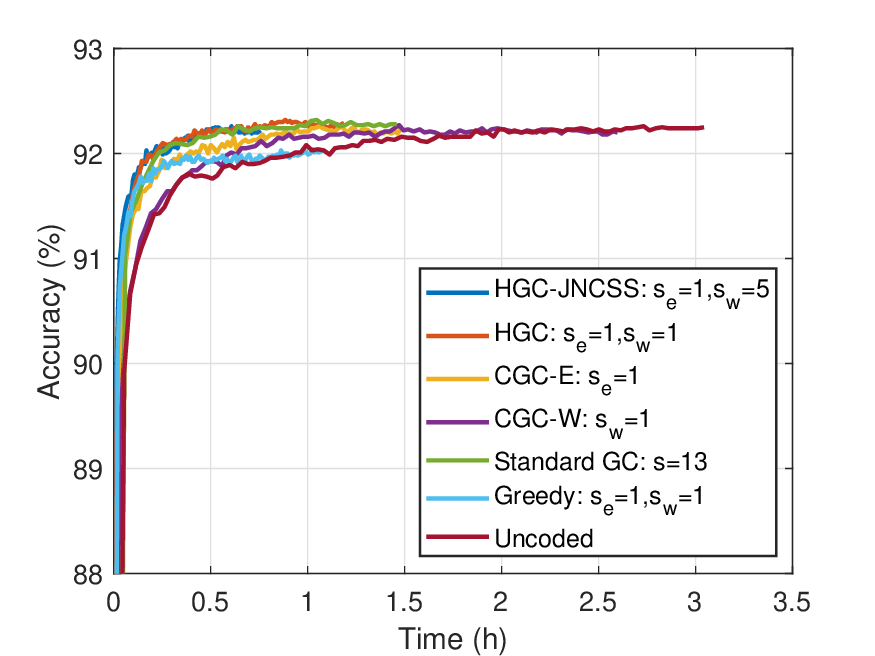}
	}
        \subfloat[MNIST and non-IID Level \uppercase\expandafter{\romannumeral2}]{
   		\label{fig4b}
            \includegraphics[width=0.325\textwidth]{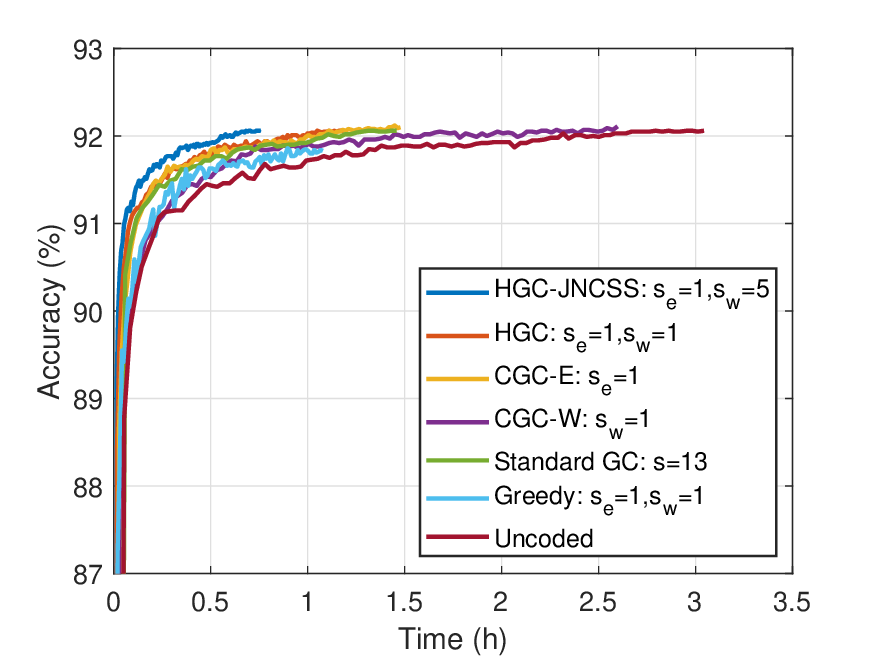}
    	}
        \subfloat[MNIST and non-IID Level \uppercase\expandafter{\romannumeral3}]{
            \label{fig4c}
	      \includegraphics[width=0.325\textwidth]{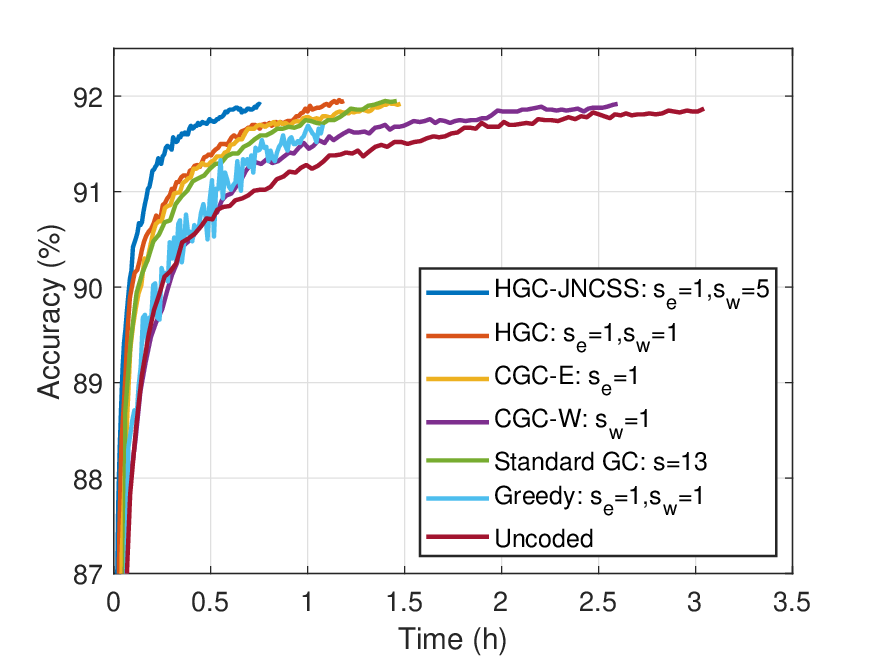}
	}
	\\
	\subfloat[CIFAR-10 and non-IID Level \uppercase\expandafter{\romannumeral1}]{
            \label{fig4d}
		\includegraphics[width=0.325\textwidth]{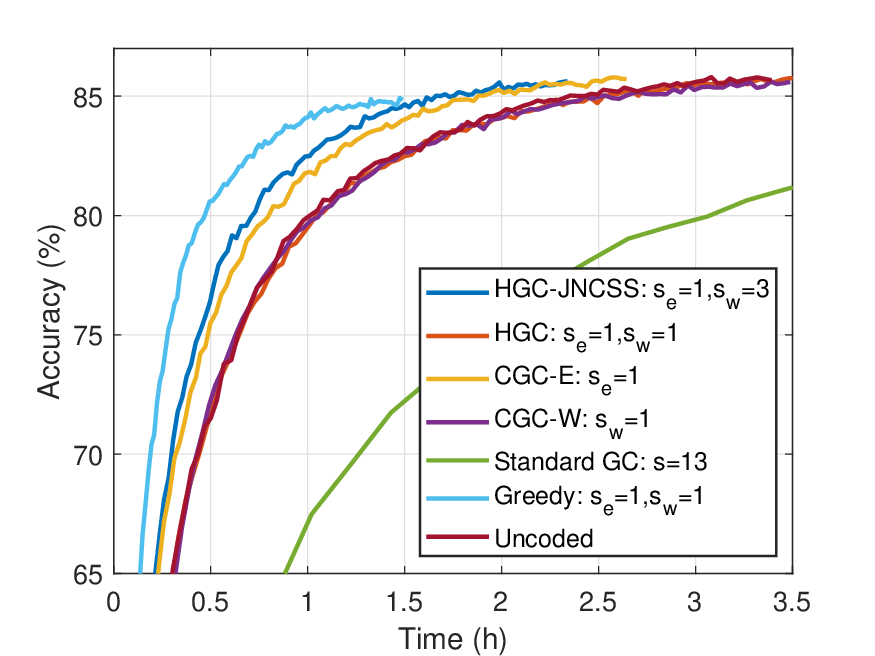}
	}
        \subfloat[CIFAR-10 and non-IID Level \uppercase\expandafter{\romannumeral2}]{
            \label{fig4e}
		\includegraphics[width=0.325\textwidth]{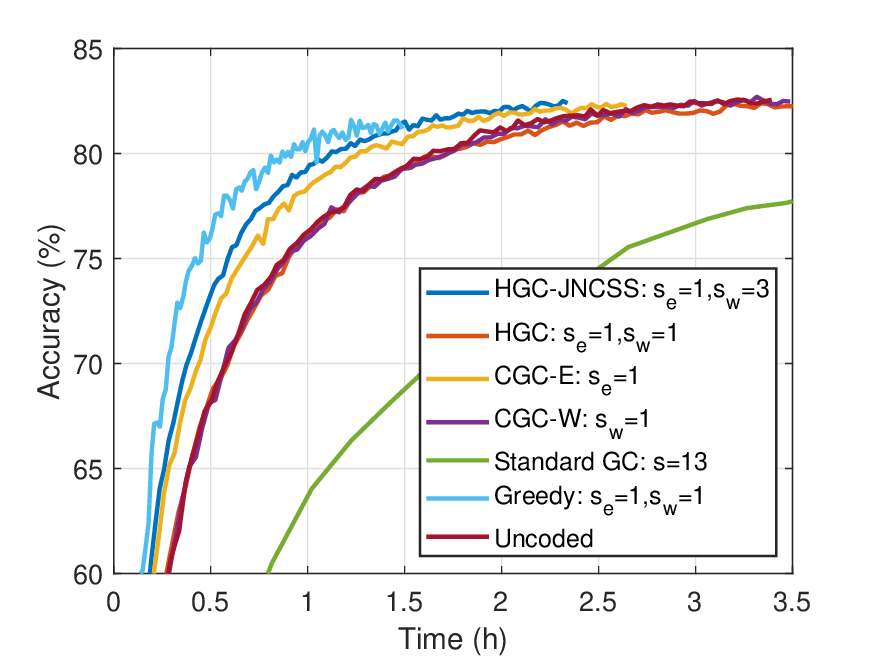}
    	}
        \subfloat[CIFAR-10 and non-IID Level \uppercase\expandafter{\romannumeral3}]{
            \label{fig4f}
		\includegraphics[width=0.325\textwidth]{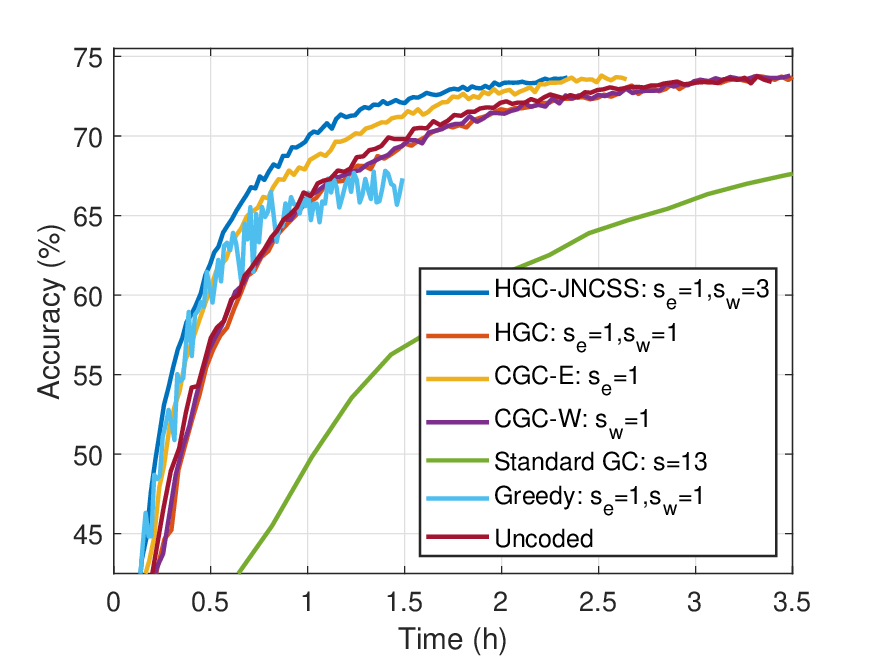}
    	}
	\caption{Test accuracy curves with respect to training time for different datasets and data non-IID levels.}
	\label{fig4}
\end{figure*}

In this section, we carry out simulations to demonstrate the performance of our proposed hierarchical gradient coding scheme and the JNCSS optimization approach. First, we describe our simulation setting, and then we present the simulation results.

\subsection{Simulation Setting}
We consider a hierarchical distributed learning system with 1 master and $n=4$ edge nodes which connect to $m=10$ workers separately. And the runtime model described in Section \ref{set51} is used for the system with the  parameters as followings. To model heterogeneity across distinct edge nodes and workers, we first introduce 3 types of edge nodes: (1) Type \uppercase\expandafter{\romannumeral1}: 1 node with strong communication capability that $p_{i}=0.1$ and $\tau_{i}=50$ ms, (2) Type \uppercase\expandafter{\romannumeral2}: 2 nodes with normal communication capability that $p_{i}=0.1$ and $\tau_{i}=100$ ms, (3) Type \uppercase\expandafter{\romannumeral3}: 1 node with weak communication capability that $p_{i}=0.2$ and $\tau_{i}=500$ ms. And the workers of each edge node are divided into 4 type: (1) Type \uppercase\expandafter{\romannumeral1}: 5 workers with strong computation and strong communication capabilities that $p_{(i,j)}=0.1$, $\tau_{(i,j)}=50$ ms and $\gamma_{(i,j)}=0.1$ $\mathrm{ms}^{-1}$, (2) Type \uppercase\expandafter{\romannumeral2}: 2 workers with strong computation but weak communication capabilities that $p_{(i,j)}=0.5$, $\tau_{(i,j)}=100$ ms and $\gamma_{(i,j)}=0.1$ $\mathrm{ms}^{-1}$, (3) Type \uppercase\expandafter{\romannumeral3}: 2 workers with weak computation but strong communication capabilities that $p_{(i,j)}=0.1$, $\tau_{(i,j)}=50$ ms and $\gamma_{(i,j)}=0.01$ $\mathrm{ms}^{-1}$, (4) Type \uppercase\expandafter{\romannumeral4}: 1 worker with weak computation and weak communication capabilities that $p_{(i,j)}=0.5$, $\tau_{(i,j)}=100$ ms and $\gamma_{(i,j)}=0.01$ $\mathrm{ms}^{-1}$.

We adopt 2 popular image classification datasets MNIST\cite{lecun1998gradient} and CIFAR-10\cite{krizhevsky2009learning} on which we train a logistic regression model and a CNN including 6 convolution layers and 3 fully connected layers, respectively. All the training samples are divided into $K=40$ disjoint sub-datasets. Moreover, we consider 3 levels of non-IID of local sub-datasets: (1) Level \uppercase\expandafter{\romannumeral1}: Each sub-dataset is allocated data samples randomly from all the sample types, (2) Level \uppercase\expandafter{\romannumeral2}: Each sub-dataset is allocated data samples randomly from 5 sample types at most, (3) Level \uppercase\expandafter{\romannumeral3}: Each sub-dataset is allocated data samples randomly from 2 sample types at most. For MNIST, workers with strong computation capability have $c_{(i,j)}=10$ ms, and other workers with weak computation capability have $c_{(i,j)}=50$ ms. For CIFAR-10, workers with strong computation capability have $c_{(i,j)}=100$ ms, and other workers with weak computation capability have $c_{(i,j)}=500$ ms.

In the simulations, conventional gradient coding schemes designed between workers and the master is not included since it is unable to decode and pre-aggregate the partial gradient results in edge nodes, resulting in higher communication cost, as described in Section \ref{set1}. Hence, we compare the following schemes, all of which can perform decoding or pre-aggregation in edge nodes.
\begin{itemize}
\item Uncoded: Each worker computes the gradient results over just one local sub-dataset. And each edge node waits to aggregate the gradient results from all the workers while the master also waits to aggregate from all the edge nodes.

\item Greedy: Each worker computes the gradient results over just one local sub-dataset, and each edge node aggregates the gradient results from the fastest $m-s_{w}$ workers while the master aggregates from the fastest $n-s_{e}$ edge nodes.

\item CGC-W: Conventional single-layer gradient coding scheme designed between workers and edge nodes to tolerate stragglers of workers.

\item CGC-E: Conventional single-layer gradient coding scheme designed between edge nodes and the master to tolerate edge stragglers.

\item Standard GC: Conventional single-layer gradient coding scheme designed in a standard worker-master architecture without edge nodes. To provide equal level of stragglers tolerance as $s_{e}$, $s_{w}$ in hierarchical architecture, Standard GC should tolerate $s=s_{e}m+s_{w}(n-s_{e})$ stragglers in worker-master architecture.

\item HGC: The proposed gradient coding scheme as described in Section \ref{set3}.

\item HGC-JNCSS: The proposed gradient coding scheme whose stragglers tolerance is optimized by Algorithm \ref{alg2}.
\end{itemize}

\subsection{Simulation Results}
The simulation results are presented in Fig. \ref{fig3} and Fig. \ref{fig4}, illustrating the performance of the above schemes for 500 training iterations. In Fig. \ref{fig3a}, Fig. \ref{fig3b}, and Fig. \ref{fig3c}, we present the accuracy curves for MNIST and different non-IID levels, with respect to training iterations. Obviously, all the coded schemes achieve comparable accuracy to the Uncoded scheme using the same number of training iterations, due to the ability to recover full gradients without stragglers. In contrast, the Greedy scheme, which ignores those stragglers, consistently performs worse than these schemes. As for the accuracy curves in terms of training time presented in Fig. \ref{fig4a}, Fig. \ref{fig4b}, and Fig. \ref{fig4c}, we can observe that our hierarchical gradient coding (HGC) scheme completes the training with significantly less time than the conventional coded and Uncoded schemes. In particular, the HGC-JNCSS scheme stands out by achieving the shortest training time among all schemes, outpacing even the Greedy scheme.

And the accuracy curves for CIFAR-10, presented in Fig. \ref{fig3d}, Fig. \ref{fig3e}, and Fig. \ref{fig3f} respectively, illustrate the performance of different schemes over the course of training iterations. Similarly, it can be observed that all the coded schemes achieve comparable accuracy with the Uncoded scheme during training, except for the Greedy scheme. However, as presented in Fig. \ref{fig4d}, Fig. \ref{fig4e}, and Fig. \ref{fig4f}, the HGC scheme is unable to complete the training faster than the conventional coded and Uncoded schemes, because of the inappropriate straggler tolerance selection, while our HGC-JNCSS scheme always provides a much better speed up. In addition, despite completing training earlier, the Greedy scheme yields lower accuracy compared to other schemes.

\begin{figure}[t]
    \centering
    \includegraphics[width=0.5\textwidth]{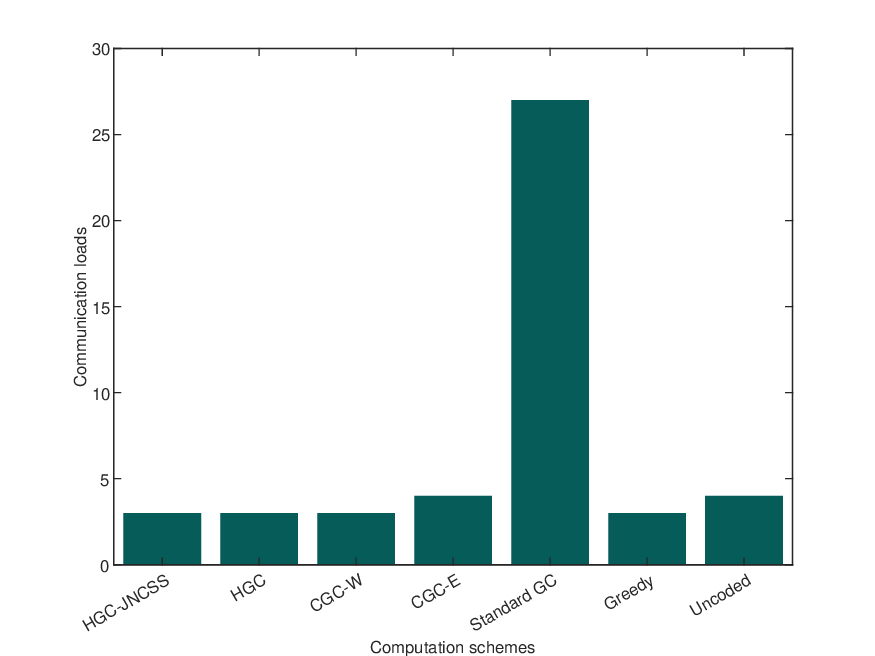}
    \caption{Communication loads of the master with different schemes.}
\label{comloads}
\end{figure}
\begin{figure}[t]
	\centering
	\includegraphics[width=0.5\textwidth]{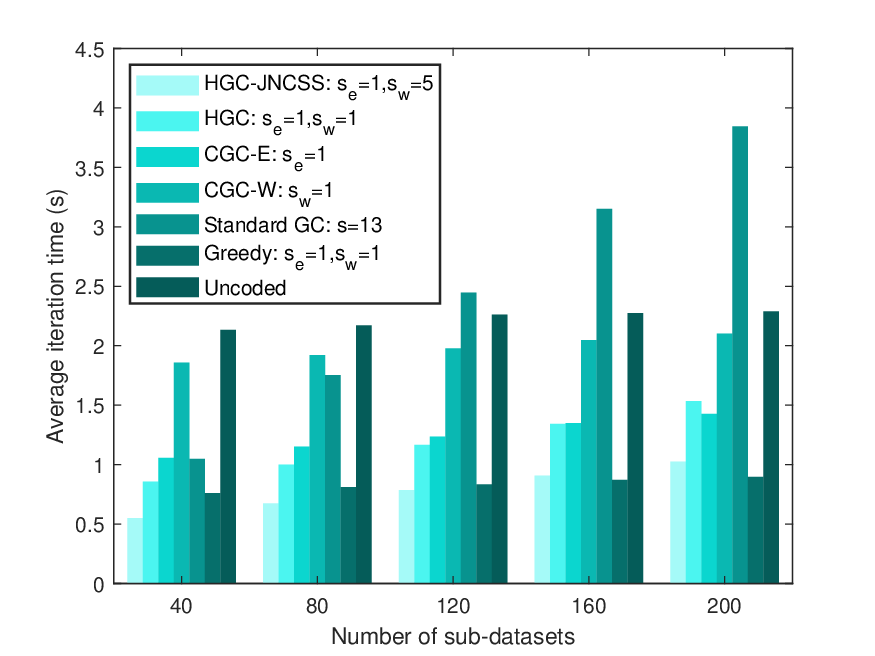}
	\caption{Average iteration times of different schemes with different numbers of sub-datasets.}
        \label{fig5}
\end{figure}
We also compare the communication loads of the master with different schemes, which is proportional to the number of computation results it receives. And the simulation results are presented in Fig. \ref{comloads}. We can observe that with the introduction of hierarchical architecture, the communication loads of the Uncoded scheme and CGC-E scheme are significantly reduced compared to the Standard GC scheme which only considers worker-master architecture. Moreover, the other schemes further reduce the communication loads due to their abilities to tolerate straggling edge nodes.

\begin{table*}[t]\label{tab1}
\caption{Training time to reach target accuracy for different datasets and data non-IID levels.}
\begin{center}
\resizebox{\textwidth}{!}{
\begin{tabular}{c|c|c c c c c c c}
\hline
Dataset&Accuracy (\%)&HGC-JNCSS (h)&HGC (h)&CGC-E (h)&CGC-W (h)&Standard GC (h)&Greedy (h)&Uncoded (h)\\
\hline
MNIST-\uppercase\expandafter{\romannumeral1} & 92.0 & \textbf{0.16} & 0.18 & 0.32 & 0.51 & 0.23 & 0.38 & 0.86\\
MNIST-\uppercase\expandafter{\romannumeral2} & 91.8 & \textbf{0.18} & 0.28 & 0.27 & 0.50 & 0.57 & 0.37 & 0.61\\
MNIST-\uppercase\expandafter{\romannumeral3} & 91.0 & \textbf{0.17} & 0.28 & 0.33 & 0.59 & 0.35 & 0.44 & 0.71\\
\hline
CIFAR-10-\uppercase\expandafter{\romannumeral1} & 84.5 & \textbf{1.43} & 2.12 & 1.60 & 2.14 & 6.12 & 1.09 & 2.04\\
CIFAR-10-\uppercase\expandafter{\romannumeral2} & 82.5 & \textbf{2.16} & 3.27 & 2.46 & 2.99 & 9.44 & — & 2.89\\
CIFAR-10-\uppercase\expandafter{\romannumeral3} & 72.5 & \textbf{1.56} & 2.34 & 1.75 & 2.32 & 6.73 & — & 2.25\\
\hline
\end{tabular}}
\label{tab1}
\end{center}
\end{table*}
To highlight the performance of our proposed schemes, especially the HGC-JNCSS scheme, we conduct a comparison of different schemes' training time required to achieve target accuracy, as presented in Table \ref{tab1}. For MNIST, both HGC and HGC-JNCSS schemes exhibit considerable acceleration compared to the other schemes. In particular, the HGC scheme shows speed-ups of up to 2.83$\times$ and 4.78$\times$ over the conventional coded and Uncoded schemes, while the HGC-JNCSS scheme demonstrates a 1.64$\times$ speed-up compared to the HGC scheme. As for CIFAR-10, the HGC-JNCSS scheme achieves speed-ups of up to 4.37$\times$ and 1.44$\times$ over the conventional coded and Uncoded schemes. Additionally, with the increasing of data non-IID level, the Greedy scheme deteriorates and even fails to achieve target accuracy.

To further demonstrate the superiority of our schemes, in Fig. \ref{fig5}, we present the average iteration times of different schemes, with the same parameters setting for MNIST and different sub-datasets numbers. The average iteration times of all the schemes increase with the number of sub-datasets raising from 40 to 200, and our proposed HGC and HGC-JNCSS schemes outperform the other coded schemes and Uncoded scheme most of the time. Wherein, the HGC scheme achieves up to $60.1\%$ and $59.8\%$ performance gain over the conventional coded and Uncoded schemes. Moreover, the HGC-JNCSS scheme achieves up to $33.7\%$ performance gain over the HGC scheme. Besides, although the Greedy scheme achieves faster average iteration time when number of sub-datasets increases, it comes with the cost of its worst model performance as shown in our other simulation results.

\section{Related Works}\label{set7}
To mitigate the stragglers effect in distributed computing systems, a series of studies introduced coding theory-based techniques. By adding redundant computing tasks to the workers, these techniques enable the central master to recover the final computing results with the outputs from non-stragglers. Lee et al.\cite{lee2017speeding} utilized MDS codes for stragglers mitigation in distributed matrix multiplication. For high-dimension matrix multiplication, Lee et al.\cite{lee2017high} constructed polynomial codes to improve the performance. Dutta et al.\cite{dutta2017coded} proposed a coded convolution scheme. And Yu et al.\cite{yu2017coded} also utilized MDS codes to speed up distributed Fourier transform.

\begin{table*}[t]\label{tab2}
\caption{Related Works on Gradient Coding.}
\begin{center}
\begin{tabular}{c c c c}
\hline
Ref. & Hierarchical Architecture & Worker Stragglers Tolerance & Edge Stragglers Tolerance\\
\hline
\cite{tandon2017gradient,raviv2018gradient,halbawi2018improving,wang2019heterogeneity,ye2018communication,cao2021adaptive} & \XSolidBrush & \Checkmark & \XSolidBrush \\
\cite{prakash2020hierarchical,sasidharan2022coded} & \Checkmark & \Checkmark & \XSolidBrush \\
This Work & \Checkmark & \Checkmark &\Checkmark \\
\hline
\end{tabular}
\label{tab2}
\end{center}
\end{table*}
Especially, Tandon et al.\cite{tandon2017gradient} first proposed gradient coding to handle the stragglers effect in distributed gradient descent. After that, Raviv et al.\cite{raviv2018gradient} and Halbawi et al.\cite{halbawi2018improving} utilized cyclic MDS codes and Reed-Solomon codes, respectively, to achieve the same recovery threshold as in\cite{tandon2017gradient}. And Wang et al.\cite{wang2019heterogeneity} proposed a heterogeneity-aware gradient coding scheme while Ye et al.\cite{ye2018communication} and Cao et al.\cite{cao2021adaptive} considered the trade-off between computational loads, stragglers tolerance and communication cost in gradient coding, to further reduce the runtime. In addition, Ozfatur et al.\cite{DSW8755563} and Buyukates et al.\cite{TCOM9755943} utilized partial computing results and clustering to improve the performance of gradient coding. Above studies on gradient coding focused on the workers-master topology design. Because of the straggling communication links, Prakash et al.\cite{prakash2020hierarchical} and Sasidharan et al.\cite{sasidharan2022coded} introduced a hierarchical architecture in gradient coding, where reliable helper nodes are located close to the workers. However, such a hierarchical design does not consider the edge stragglers, which may lead to more severe stragglers effect. Related works on gradient coding are summarized in Table \ref{tab2}.

As one of the important frameworks to distributedly train machine learning model, federated learning has also been combined with coding techniques in recent studies, such as CodedFedL\cite{prakash2020coded}, CodedPaddedFL\cite{schlegel2021codedpaddedfl} and DRes-FL\cite{shao2022dres}.

There are also some studies considering the runtime modeling and optimization of coded distributed computing systems. Van Huynh et al.\cite{van2021joint} and Nguyen et al.\cite{nguyen2021jointly} formulated coding scheme optimization problems, then solve with DRL and MINLP solvers, respectively. However, most of these studies focus on the modeling and optimization of conventional workers-master topology designed systems, but not hierarchical architecture.

It is also noting that joint communication and computation design for distributed edge learning has been extensively investigated in the literature, e.g., \cite{jpba9194337,efrra9145118,effel9475121,jtcro9682096,hcra9745059}. In particular, many works focused on the context of federated edge learning. For example, Huang et al.\cite{efrra9145118} optimized ratio resource allocation while Mo et al.\cite{effel9475121} proposed a joint communication and computation design to achieve energy-efficiency. Besides, Huang et al.\cite{jtcro9682096} jointly optimized communication topology and computation resource in federated edge learning.

\section{Conclusion}\label{set8}
In this work, we first derive the fundamental computational trade-off between the computational loads of workers and the stragglers tolerance. We then propose a hierarchical gradient coding framework to provide better stragglers mitigation for more severe stragglers effect in hierarchical architecture, which achieves the optimal trade-off. Moreover, to further improve the performance of our proposed framework in heterogeneous scenarios, we formulate a jointly node and coding scheme selection problem with the objective of minimizing the expected execution time for each iteration in hierarchical coded distributed learning systems, and develop an optimization algorithm to mathematically solve the problem. Simulation results demonstrate the superiority of our schemes compared with the other conventional schemes. Interesting directions for future work include generalization of our proposed scheme to the unbalanced training workload allocation and exploring the trade-off between computational loads, stragglers tolerance, and communication cost in hierarchical distributed learning systems.

\begin{appendices}
\section{Proof of Theorem \ref{theo2}}\label{apc}
We begin with considering a sub-problem of the JNCSS problem $\mathcal{P}_{1}$, as following:
\begin{equation}
    \begin{aligned}
        \mathcal{P}_{2}:&\min_{\boldsymbol{e},\boldsymbol{w}}\ \ T_{tol},\\
        &\ \ \text{s.t.}\ \ (\ref{opt1})-(\ref{opt6}).
    \end{aligned}
\end{equation}
For any fixed stragglers tolerance $(s_{e},s_{w})$, we denote the optimal value of $\mathcal{P}_{2}$ as $T_{tol}(s_{e},s_{w})$, and the optimal value of $\mathcal{P}_{1}$ can be expressed as $\min\limits_{s_{e},s_{w}}T_{tol}(s_{e},s_{w})$. According to Algorithm \ref{alg2} in Section \ref{set53}, we can determine that $\hat{T}_{tol}(s_{e},s_{w})=\min\limits_{(n-s_{e})-\mathrm{th}}(A_{i}+\min\limits_{(m_{i}-s_{w})-\mathrm{th}}B_{(i,j)})$. Then, we try to prove that $\hat{T}_{tol}(s_{e},s_{w})$ equals to the optimal $T_{tol}(s_{e},s_{w})$ in $\mathcal{P}_{2}$, for any fixed stragglers tolerance $(s_{e},s_{w})$. Without loss of generality, we suppose that $B_{(i,j_{1})}\leq B_{(i,j_{2})}$ if $j_{1}\leq j_{2}$, and $A_{i_{1}}+\min\limits_{(m_{i_{1}}-s_{w})-\mathrm{th}}B_{(i_{1},j)}\leq A_{i_{2}}+\min\limits_{(m_{i_{2}}-s_{w})-\mathrm{th}}B_{(i_{2},j)}$ if $i_{1}\leq i_{2}$. With these assumptions, we can rewrite $\hat{T}_{tol}(s_{e},s_{w})$ as
\begin{equation}
    \hat{T}_{tol}(s_{e},s_{w})=A_{n-s_{e}}+B_{(n-s_{e},m_{n-s_{e}}-s_{w})}.
\end{equation}
And one of the corresponding nodes selection schemes, which achieve $\hat{T}_{tol}(s_{e},s_{w})$, is as follows: $e_{i}=1$ if $i\textless n-s_{e}$, and $w_{(i,j)}=1$ if $e_{i}=1$, $j\textless m_{j}-s_{w}$.

If there exists a more efficient edge nodes selection scheme resulting in a lower value of $T_{tol}$ than $\hat{T}_{tol}(s_{e},s_{w})$, it implies that at least one edge node satisfies
\begin{equation}
A_{i}+B_{(i,m_{i}-s_{w})}\textless A_{n-s_{e}}+B_{(n-s_{e},m_{n-s_{e}}-s_{w})},\ \ i\textgreater n-s_{e}.
\end{equation}
Obviously, there is a contradiction.

Similarly, if there exists a more efficient workers selection scheme while edge nodes selection scheme is the same, it implies that at least one worker satisfies
\begin{equation}
B_{(n-s_{e},j)}\textless B_{(n-s_{e},m_{n-s_{e}}-s_{w})},\ \ j\textgreater m_{n-s_{e}}-s_{w},
\end{equation}
which is also a contradiction. By the proof above, we have proved that for any fixed $(s_{e},s_{w})$, $\hat{T}_{tol}(s_{e},s_{w})=\min\limits_{(n-s_{e})-\mathrm{th}}(A_{i}+\min\limits_{(m_{i}-s_{w})-\mathrm{th}}B_{(i,j)})$ equals to the optimal $T_{tol}(s_{e},s_{w})$ of $\mathcal{P}_{2}$. And according to Line 11 of Algorithm \ref{alg2}, $\hat{T}_{tol}=\min\limits_{s_{e},s_{w}}\hat{T}_{tol}(s_{e},s_{w})$ is the optimal value of the JNCSS problem $\mathcal{P}_{1}$. Consequentially, the corresponding nodes selection and coding scheme return of Algorithm \ref{alg2} is also one of the optimal schemes.

\section{Proof of Theorem \ref{t2}}\label{apa}
Before proving Theorem \ref{t2}, we first introduce the following lemma, which is proved in Appendix \ref{apb}.
\begin{lemma}\label{lem1}
If $n$ random variables $X_{i}$, $i\in[n]$, with expectations $u_{i}$ and variances $\sigma_{i}^{2}$, are arranged in order of magnitude as $X_{(1)}\leq X_{(2)}\leq ...\leq X_{(n)}$, we call $X_{(i)}$ the $i$-th order statistic\cite{david2004order} with expectations $u_{(i)}$. And We have
\begin{equation}\label{lem}
    |u_{(r)}-u_{r}|\leq f(n,r)\sqrt{\sum_{i}\left[\sigma_{i}^{2}+(u_{i}-\overline{u})^{2}\right]-n\overline{\sigma}^{2}},
\end{equation}
where $f(n,r)=\sqrt{\frac{r-1}{n(n-r+1)}}+\sqrt{\frac{n-r}{nr}}$, $r\in[n]$, $\overline{u}$ and $\overline{\sigma}^{2}$ are the expectation and variance of $\overline{X}=\frac{\sum_{i}X_{i}}{n}$, respectively.
\end{lemma}
According to Algorithm \ref{alg2} and our formulation in Section \ref{set51}, we have
\begin{equation}\label{t21}
\begin{aligned}
    \mathbb{E}&\left[T_{tol}-\hat{T}_{tol}\right]=\mathbb{E}\bigg[\min\limits_{(n-\hat{s}_{e})-\mathrm{th}}(T^{i}_{com}+\min\limits_{(m_{i}-\hat{s}_{w})-\mathrm{th}}T^{(i,j)}_{tol})\\
    &-\min\limits_{(n-\hat{s}_{e})-\mathrm{th}}\left(\mathbb{E}\left[T^{i}_{com}\right]+\min\limits_{(m_{i}-\hat{s}_{w})-\mathrm{th}}\mathbb{E}\left[T^{(i,j)}_{tol}\right]\right)\bigg].
\end{aligned}
\end{equation}
Using Lemma \ref{lem1}, we have
\begin{equation}
\begin{aligned}
    \mathbb{E}\left[\min\limits_{(m_{i}-\hat{s}_{w})-\mathrm{th}}T^{(i,j)}_{tol}\right]&-f(m_{i},m_{i}-\hat{s}_{w})\Delta_{w}^{i}
    \\ & \leq\min\limits_{(m_{i}-\hat{s}_{w})-\mathrm{th}}\mathbb{E}\left[T^{(i,j)}_{tol}\right].
\end{aligned}
\end{equation}
And recall that $T^{i}_{tol}=T^{i}_{com}+\min\limits_{(m_{i}-\hat{s}_{w})-\mathrm{th}}T^{(i,j)}_{tol}$, we can upper bound (\ref{t21}) as
\begin{equation}
\begin{aligned}
    \mathbb{E}\left[T_{tol}-\hat{T}_{tol}\right]\leq\mathbb{E}&\bigg[\min\limits_{(n-\hat{s}_{e})-\mathrm{th}}T^{i}_{tol}-\min\limits_{(n-\hat{s}_{e})-\mathrm{th}}\mathbb{E}\left[T^{i}_{tol}\right]
    \\& +\max\limits_{i}f(m_{i},m_{i}-\hat{s}_{w})\Delta_{w}^{i}\bigg].
\end{aligned}
\end{equation}
Using Lemma \ref{lem1} again, we can upper bound the inequality above as
\begin{equation}\label{t22}
    \mathbb{E}\left[T_{tol}-\hat{T}_{tol}\right]\leq f(n,n-\hat{s}_{e})\Delta_{e}+\max\limits_{i}f(m_{i},m_{i}-\hat{s}_{w})\Delta_{w}^{i}.
\end{equation}

For the same reason, we can also prove that $\mathbb{E}\left[\hat{T}_{tol}-T_{tol}\right]\leq f(n,n-\hat{s}_{e})\Delta_{e}+\max\limits_{i}f(m_{i},m_{i}-\hat{s}_{w})\Delta_{w}^{i}$. With the inequalities above, Theorem \ref{t2} is proved.

\section{Proof of Lemma \ref{lem1}}\label{apb}
For any $n$ constants $v_{1},v_{2},...,v_{n}$, by Cauchy's inequality, we have
\begin{equation}\label{le1}
\begin{aligned}
    \left|\sum_{i}{v_{i}(u_{(i)}-\overline{u})}\right|&=\left|\mathbb{E}\sum_{i}{(v_{i}-\overline{v})(X_{(i)}-\overline{X})}\right|\\
    &\leq\left[\sum_{i}{(v_{i}-\overline{v})^{2}}\right]^{\frac{1}{2}}\left[\mathbb{E}\sum_{i}{(X_{(i)}-\overline{X})^{2}}\right]^{\frac{1}{2}}\\
    &=\left[\sum_{i}{(v_{i}-\overline{v})^{2}}\right]^{\frac{1}{2}}\left[\mathbb{E}\sum_{i}{(X_{i}-\overline{X})^{2}}\right]^{\frac{1}{2}},
\end{aligned}
\end{equation}
where $\overline{v}=\frac{\sum_{i}{v_{i}}}{n}$. For the right hand side of (\ref{le1}), we have
\begin{equation}
    \begin{aligned}
        \sum_{i}{(X_{i}-\overline{X})^{2}}=& \sum_{i}{(X_{i}-u_{i})^{2}}+\sum_{i}(u_{i}-\overline{u})^{2}
        \\& +2\sum_{i}{(X_{i}-u_{i})(u_{i}-\overline{u})}-n(\overline{X}-\overline{u})^{2}.
    \end{aligned}
\end{equation}
Applying equation above to the right-hand side of (\ref{le1}), it becomes
\begin{equation}\label{le2}
\begin{aligned}
    \left|\sum_{i}{v_{i}(u_{(i)}-\overline{u})}\right|\leq&\left[\sum_{i}{(v_{i}-\overline{v})^{2}}\right]^{\frac{1}{2}}
    \\ &\times\left[\sum_{i}\left[\sigma_{i}^{2}+(u_{i}-\overline{u})^{2}\right]-n\overline{\sigma}^{2}\right]^{\frac{1}{2}}.
\end{aligned}
\end{equation}
Similarly, we can also prove that
\begin{equation}\label{le3}
\begin{aligned}
   \left|\sum_{i}{v_{i}(u_{i}-\overline{u})}\right|\leq&\left[\sum_{i}{(v_{i}-\overline{v})^{2}}\right]^{\frac{1}{2}}
   \\ &\times\left[\sum_{i}\left[\sigma_{i}^{2}+(u_{i}-\overline{u})^{2}\right]-n\overline{\sigma}^{2}\right]^{\frac{1}{2}}.
\end{aligned}
\end{equation}

Moreover, it is easy to get that
\begin{equation}\label{le4}
    u_{(r)}\leq \frac{\sum^{n}_{i=r}u_{(i)}}{n-r+1},
\end{equation}
\begin{equation}\label{le5}
    u_{r}\geq\frac{\sum^{r}_{i=1}u_{i}}{r}.
\end{equation}
Then, applying (\ref{le2}) to the right hand side of inequality (\ref{le4}) above by setting $v_{i}=\frac{1}{n-r+1}$ when $i\geq r$ else $v_{i}=0$, we have
\begin{equation}\label{le6}
\begin{aligned}
    u_{(r)}\leq\sqrt{\frac{r-1}{n(n-r+1)}}\left[\sum_{i}\left[\sigma_{i}^{2}+(u_{i}-\overline{u})^{2}\right]-n\overline{\sigma}^{2}\right]^{\frac{1}{2}}+\overline{u}.
\end{aligned}
\end{equation}
By setting $v_{i}=\frac{1}{r}$ when $i\leq r$ else $v_{i}=0$, applying (\ref{le3}) to to the right hand side of inequality (\ref{le5}), we have
\begin{equation}\label{le7}
\begin{aligned}
    u_{r}\geq\overline{u}-\sqrt{\frac{n-r}{nr}}\left[\sum_{i}\left[\sigma_{i}^{2}+(u_{i}-\overline{u})^{2}\right]-n\overline{\sigma}^{2}\right]^{\frac{1}{2}}.
\end{aligned}
\end{equation}
According to inequalities (\ref{le6}) and (\ref{le7}), we can prove that
\begin{equation}
    u_{(r)}-u_{r}\leq f(n,r)\sqrt{\sum_{i}\left[\sigma_{i}^{2}+(u_{i}-\overline{u})^{2}\right]-n\overline{\sigma}^{2}}.
\end{equation}
For the same reason, we can also get $u_{r}-u_{(r)}\leq f(n,r)\sqrt{\sum_{i}\left[\sigma_{i}^{2}+(u_{i}-\overline{u})^{2}\right]-n\overline{\sigma}^{2}}$. Hence, we have proved Lemma \ref{lem1}.
\end{appendices}

\bibliographystyle{IEEEtran}
\bibliography{IEEEabrv,mylib}

\end{document}